%% file: submit_1115.tex
\documentclass[reqno,oneside,11pt]{amsart}
\input{headers.tex}
\input{defs.tex}

\usepackage{setspace}
\usepackage{bbm,multirow,longtable,rotating,array,booktabs,float,graphicx}
\usepackage{hyperref}

\title[State-domain Change Point Detection]{State-domain Change Point Detection for Nonlinear Time Series Regression}
\author{Yan Cui$^{1,2}$}
\author{Jun Yang$^{3}$}
\author{Zhou Zhou$^1$}
\address{$^1$ Department of Statistical Sciences,
	University of Toronto, Canada}
\address{$^2$ Institute for Advanced Study in Mathematics, Harbin Institute of Technology, China}
\address{$^3$ Department of Statistics, University of Oxford, United Kingdom}
\email{\{cui,zhou\}@utstat.toronto.edu} \email{jun.yang@stats.ox.ac.uk}

\begin{document}	
	\begin{abstract}
		\input{changepoint-abstract.tex}
		
	\end{abstract}
	
	\maketitle


\input{body_1115.tex}

\bibliographystyle{agsm}
\bibliography{time_series_bib}

\clearpage
\appendix


\end{document}

%% file: headers.tex
\usepackage[foot]{amsaddr}
\usepackage[utf8]{inputenc} 
\usepackage{amsmath, amssymb,bm, cases, mathtools, thmtools}
\usepackage{caption}
\usepackage{subfigure}
\usepackage{verbatim}
\usepackage{graphicx}\graphicspath{{figures/}}
\usepackage{multicol}
\usepackage{tabularx}
\usepackage[usenames,dvipsnames]{xcolor}
\usepackage{mathrsfs} 
\usepackage{url}
\usepackage{bbm}
\usepackage{enumitem}
\usepackage{natbib}

\usepackage[colorlinks,citecolor=blue,urlcolor=blue,linkcolor=RawSienna]{hyperref}
\usepackage{hypernat}
\usepackage{datetime}

\DeclareMathAlphabet\EuRoman{U}{eur}{m}{n}
\SetMathAlphabet\EuRoman{bold}{U}{eur}{b}{n}

\usepackage{euscript,microtype}

\usepackage[capitalize]{cleveref}

\crefname{lemma}{Lemma}{Lemmas}
\crefname{corollary}{Corollary}{Corollaries}
\crefname{theorem}{Theorem}{Theorems}

\makeatletter
\let\reftagform@=\tagform@
\def\tagform@#1{\maketag@@@{\ignorespaces\textcolor{gray}{(#1)}\unskip\@@italiccorr}}
\renewcommand{\eqref}[1]{\textup{\reftagform@{\ref{#1}}}}
\makeatother

\input{commenting.tex}

\declaretheorem[style=plain,numberwithin=section,name=Theorem]{theorem}

\declaretheorem[style=plain,sibling=theorem,name=Proposition]{proposition}
\declaretheorem[style=plain,sibling=theorem,name=Corollary]{corollary}

\declaretheorem[style=remark,qed=$\triangleleft$,sibling=theorem,name=Remark]{remark}

\numberwithin{theorem}{section}



%% file: commenting.tex
\definecolor{WowColor}{rgb}{.75,0,.75}
\definecolor{SubtleColor}{rgb}{0,0,.50}

\newcounter{margincounter}

%% file: defs.tex
\def\[#1\]{\begin{align}#1\end{align}}
\def\*[#1\]{\begin{align*}#1\end{align*}}

\newcommand{\bigO}{\mathcal{O}}
\newcommand{\Cov}{\textrm{Cov}}

\newcommand{\tK}{\tilde{K}}

\newcommand{\Reals}{\mathbb{R}}

\newcommand{\dee}{\mathrm{d}}

\DeclareMathOperator*{\newlim}{\mathrm{lim}\vphantom{\mathrm{infsup}}}
\DeclareMathOperator*{\newmin}{\mathrm{min}\vphantom{\mathrm{infsup}}}
\DeclareMathOperator*{\newmax}{\mathrm{max}\vphantom{\mathrm{infsup}}}
\DeclareMathOperator*{\newinf}{\mathrm{inf}\vphantom{\mathrm{infsup}}}
\DeclareMathOperator*{\newsup}{\mathrm{sup}\vphantom{\mathrm{infsup}}}
\renewcommand{\lim}{\newlim}
\renewcommand{\min}{\newmin}
\renewcommand{\max}{\newmax}
\renewcommand{\inf}{\newinf}
\renewcommand{\sup}{\newsup}

\renewcommand{\Pr}{\mathbb{P}}
\def\EE{\mathbb{E}}



%% file: changepoint-abstract.tex
Change point detection in time series has attracted substantial interest, but most of the existing results have been focused on detecting change points in the time domain. This paper considers the situation where nonlinear time series have potential change points in the state domain. We apply a density-weighted anti-symmetric kernel function to the state domain and therefore propose a nonparametric procedure to test the existence of change points. When the existence of change points is affirmative, we further introduce an algorithm to estimate the number of change points together with their locations. Theoretical results of the proposed detection and estimation procedures are given and a real dataset is used to illustrate our methods.

%% file: body_1115.tex
	\section{Introduction}

Consider the following state-domain nonlinear auto-regression
\[
X_i=\mu(X_{i-1})+\epsilon_i,\label{new-model}
\]
where $\mu(\cdot)$ is an unknown regression function, $\{\epsilon_i\}$ is a martingale difference sequence such that $\EE[\epsilon_i\mid(\epsilon_{i-1},\epsilon_{i-2},\cdots)]=0 $ almost surely. 
Special cases of \cref{new-model} include threshold  AR models \citep{Tong1990}, exponential AR models \citep{Hagg1981} and ARCH models \citep{Engle1982}, among others. Furthermore, \cref{new-model} can be viewed as a discretized version of the diffusion model
\[
\dee X_t=\mu(X_t)\dee t+\dee \mathbb{M}(t),\label{model2}
\]
where $\mu(\cdot)$ is the instantaneous return or drift function, and $\{\mathbb{M}(t)\}$ is a continuous-time martingale. In the literature, the special case of Model \eqref{model2} with $\dee \mathbb{M}(t)=\sigma(X_t)\dee \mathbb{B}(t)$ has been widely discussed to understand and model nonlinear temporal systems in economics and finance, where $\mathbb{B}(t)$ denotes the standard Brownian motion and $\sigma^2(\cdot)$ is understood as the volatility function. Among others, \cite{Stan1997}, \cite{CP2000} and \cite{Fan2003} considered the nonparametric estimation of $\mu(\cdot)$ and $\sigma^2(\cdot)$. \cite{Zhao2011} addressed the model validation problem for \cref{model2}. In particular, \cref{model2} can be used to model the temporal dynamics of financial data with $\{X_t\}$ being interest rates, exchange rates, stock prices or other economic quantities. Among others, \cite{Zhao2006} considered kernel quantile estimates of \cref{model2} for the exchange rates between Pound and USD. \cite{Liu2010a} constructed simultaneous confidence bands for $\mu(\cdot)$ and $\sigma(\cdot)$ with the U.S. Treasury yield curve rates data. See also the latter papers for further references. Observe that we allow the error process to be general martingale differences in \eqref{new-model} which significantly expands the applicability of our theory and methodology in economic applications. As pointed out by one referee, conditional moment restrictions in dynamic economic models routinely arise from Euler/Bellman equations in dynamic programming, which are martingale differences. Furthermore, asset returns, due to the no-arbitrage theory, are (semi)martingales. Hence, their (demeaned) returns are martingale differences.

Throughout this article, following Chapter 6.3 of \cite{Fan03}, we shall call \eqref{new-model} a state-domain nonlinear regression model. The term ``state domain'' originated from the celebrated state-space models (e.g. \citet{kalman1960new} and  \citet[Chapter 6]{shumway2000time}) where the dynamics of a sequence of state variables ($\{X_i\}$ in \cref{new-model}) are driven by a group of control variables ($\epsilon_{i}$ in \cref{new-model}) through the nonlinear state equation \eqref{new-model}. Therefore in this article the term ``state domain'' refers to the Euclidean space in which the variables on the axes are the state variables. Observe that the state-domain nonlinear regression \eqref{new-model} aims to characterize the relationship between $X_i$ and past values (states) of the time series  through a discretized stochastic differential equation. On the contrary, time-domain nonlinear regression (see e.g. \cite{Fan03}, Chapter 6.2)
\[
X_i=f(i/n)+\varepsilon_i, \quad i=1,2,\cdots, n\label{time-domain}
\]
with $\EE[\varepsilon_i]=0$ describes the relationship between $X_i$ and time.   

To date, most investigations on the nonparametric inference procedure of \cref{new-model} are based on the assumption that the underlying regression function $\mu(\cdot)$ is continuous, which may cause serious restrictions in many real applications. In fact, in parametric modeling of nonlinear time series, various choices of $\mu(\cdot)$ with possible discontinuities have drawn much attention in the literature. One of the most prominent examples is the threshold model proposed by \cite{Tong1980}, in which regime switches are triggered by an observed variable crossing an unknown threshold. Also, AR model with regime-switch controlled by a Markov chain mechanism was introduced by \cite{Tong1990}. In economics, the expanding phase and contracting phase are not always governed by the same dynamics, see \cite{Tiao1994,DJ1995,MP2000} and other references therein. As a result, the occurrence of abrupt changes in the state-domain regression function $\mu(\cdot)$ is common and detecting as well as estimating them are of vital importance. Motivated by this, in the current paper we focus on the situation where the regression function $\mu(\cdot)$ is piece-wise smooth on an interval of interest $T=[l,u]$ with a finite but unknown number of change points. More precisely, there exist $l=a_0<a_1<\cdots<a_M<a_{M+1}=u$ such that $\mu(\cdot)$ is smooth on each of the intervals $[a_0,a_1),\cdots,[a_M,a_{M+1}]$; that is, on the interval $[l,u]$
\[
\mu(x)=\sum_{j=0}^M\mu_j(x)\mathbbm{1}(a_j\le x <a_{j+1}), \label{e2}
\]
where $M$ is the total number of change points. Throughout this article, we assume $M$ is fixed. 

To our knowledge, there exist no results on change point detection of the state-domain regression function $\mu(\cdot)$ in the literature. The purpose of this paper is twofold. First we want to test whether $\mu(x)$ is smooth or discontinuous on the interval $[l,u]$; that is to test the null hypothesis $H_0: M=0$ of \cref{e2}. By sliding a density-weighted anti-symmetric kernel through the state domain, we shall suggest a nonparametric test statistic and non-trivially apply the discretized multivariate Gaussian approximation result of \cite{Zaitsev1987} to establish its asymptotic distribution. Additionally, the Gaussian approximation results also directly suggest a finite sample simulation-based bootstrapping method which improves the accuracy of the test in practical implementations. Second, if $M\ge 1$, we reject the null hypothesis and subsequently want to locate all the change points. In this case, we propose an estimation procedure and establish the corresponding asymptotic theory on the accuracy of the estimators. Finally, the above theoretical results are of general interest and could be used for a wider class of state-domain change point detection problems.

There is long-standing literature in statistics discussing jump detection of the time-domain nonlinear regression model \eqref{time-domain} where occasional jumps occur in an otherwise smoothly changing time trend $f(\cdot)$. It is impossible to show a complete reference here and we only list some representative works. \cite{Muller1992} and \cite{ES1994} employed a kernel method to estimate jump points in smooth curves. \cite{Wang1995} suggested using wavelets and provided a review of jump-point estimation. Two-step methods were considered by \cite{Muller1997} and \cite{Gijbels1999} to study the asymptotic convergence properties of the jumps. Later, \cite{Gijbel2007} suggested a compromise estimation method which can preserve possible jumps in the curve. \cite{Zhang2016} considered the situation where the trend function allows a growing number of jump points. 
In econometrics, there is a significant body of literature discussing time-domain jump detection in jump diffusion models; see for instance \cite{bollerslev2008risk,jiang2008testing,lee2012jumps} and the references therein.
On the other hand, it is well known that state-domain asymptotic theory is very different from that of the time domain (see, for instance \cite{Fan03}, Chapter 6). In our specific case, uniform asymptotic behavior of our test statistic on $[l,u]$ is arguably more difficult to establish than the corresponding problem in the time domain. 
In the current paper, we establish that, unlike time-domain change point detection problems of \eqref{time-domain} where the long-run variances of the process are of crucial importance in the asymptotics, state-domain change point detection theory of \eqref{new-model} heavily depends on the conditional variances and densities of the process $\{X_i\}$. We also provide an estimation procedure using simulated critical values to detect and locate the change points. We show that, when the jump sizes have a fixed and positive lower bound, the method will asymptotically detect all the change points with a preassigned probability and  an accuracy $c_n$ which is much smaller than $1/\sqrt{n}$, where $n$ is the length of the time series.
	
	The rest of the paper is organized as follows. In \cref{sec2}, we introduce the model framework and some basic assumptions. \cref{sec3} contains our main results, including a nonparametric test to determine the existence of change points and a procedure for estimating the number of change points together with their locations. Practical implementation based on a bootstrap procedure and a suitable method for bandwidth selection are discussed in \cref{sec4}. \cref{sec5} reports some simulation studies. A real data application of daily COVID-19 infections in Germany is carried out in \cref{sec6}. \cref{sec7} contains all the proofs of the theoretical results in \cref{sec3}.

\section{Model Formulation and Basic Assumptions}\label{sec2}

Throughout this paper, we use the following notations. A random vector $X \in\mathcal{L}^p$ if $\|X\|_p:=(\EE|X|^p)^{1/p}<\infty$. For two random variables $U$ and $V$, $F_{U\,|\,V}(\cdot)$ denotes the conditional distribution function of $U$ given $V$ and $f_{U\,|\,V}(\cdot)$ denotes the conditional density. Furthermore, for function $g$ with $\EE|g(U)|<\infty$, we let $\EE(g(U)\,|\,V):=\int g(x)\dee F_{U\,|\,V}(x)$ be the conditional expectation of $g(U)$ given $V$. Finally, $\mathbbm{1}$ stands for the indicator function.

Assume that the process $\{\epsilon_i\}$ is stationary and causal. Following \cite{Wu2005}, we assume that $\{\epsilon_i\}$ is a Bernoulli shift process such that
\[
\epsilon_i=G^*(\xi_i),\label{xi}
\]
where the function $G^*$ is a measurable function such that the process $\{\epsilon_i\}$ exists and $\xi_i=(\cdots,\eta_{i-1},\eta_i)$ is a shift process, where $\{\eta_i\}$ are independent and identically distributed (i.i.d.) random variables. Furthermore, $\{\epsilon_i\}$ is a martingale difference sequence satisfying $\EE[\epsilon_i\mid (\epsilon_{i-1},\epsilon_{i-2},\cdots)]=0$ almost surely. From \cref{xi}, one can interpret the transform $G^*$ as the underlying physical mechanism, with $\xi_i$ and $G^*(\xi_i)$ being the input and output of the system, respectively. 

Similarly, we assume
	\[
X_i=G(\xi_i),
	\]
	where $G$ is a measurable function such that $X_i$ exists. To facilitate the main results, we first introduce the time series dependence measures in \cite{Wu2005} associated with $X_i$ and $\epsilon_i$. Assume $X \in\mathcal{L}^p$, and let
\[
X_n'=G(\xi_n'), \quad \xi_n':=(\xi_{-1},\eta_0',\eta_1,\dots,\eta_n),
\]
where $X_n'$ is a coupled process of $X_n$ with $\eta_0$ replaced by an i.i.d.~copy $\eta_0'$. Then, we define the physical dependence measures of $X_i$ as
\[
\theta_{n,p}=\|X_n-X_n'\|_p.
\]
Let $\theta_{n,p}=0$ if $n<0$. Thus for $n\ge 0,~\theta_{n,p}$ measures the dependence of the output $G(\xi_n)$ on the single input $\eta_0$. We refer to \cite{Wu2005} for more details on the physical dependence measures.

Similarly, we define the physical dependence measures for the errors as
\[
\theta^*_{n,p}=\|\epsilon_n-\epsilon_n'\|_p,
\]
where $\epsilon_n'=G^*(\xi_n')$. Let $\theta^*_{n,p}=0$ if $n<0$.

Suppose that $\{X_i\}_{i=1}^n$ is observed. Recall $H_0: M=0$ and we aim to test the null hypothesis that the regression function is smooth. To this end, we introduce a density-weighted anti-symmetric kernel function $\tilde{K}_n$, which is defined by

\[
\tilde{K}_n(X,x,b):=\frac{w_n^*(x,b)K\left(\frac{X-x}{b}\right)-w_n(x,b)K^*\left(\frac{X-x}{b}\right)}{w_n(x,b)w_n^*(x,b)},\label{tildek}
\] 
where $K(\cdot)$ is a kernel function supported on $S=[0,1]$ with $\int_S K(u)\dee u=1$ and $K^*(u):=K(-u)$. The data-dependent weights $w_n(x,b)$ and $w_n^*(x,b)$ are defined by
\[
w_n(x,b):=\frac{1}{nb}\sum_{i=1}^n K\left(\frac{X_i-x}{b}\right),\quad w_n^*(x,b):=\frac{1}{nb}\sum_{i=1}^n K^*\left(\frac{X_i-x}{b}\right),
\]
where $b=b_n$ is the bandwidth satisfying $b\to 0$ and $nb\to\infty$. Note that $w_n(x,b)$ and $w_n^*(x,b)$ are one-sided kernel density estimators. Hence $\tilde{K}_n(X,x,b)$ can be approximated by $[K(\frac{X-x}{b})-K^\ast(\frac{X-x}{b})]/f(x)$, where $f(x)$ is the density function of $X_i$. Observing that $K(x)-K^\ast(x)$ is an anti-symmetric function, we then call $\tilde{K}_n(X,x,b)$ a density-weighted anti-symmetric kernel function. By sliding this kernel function $\tilde{K}_n$ through the state domain, we are able to test whether $\mu(x)$ has change points. More specifically, the quantity $\sum_{k=2}^n\tilde{K}_n(X_{k-1},x,b)X_k/{nb}$ is a boundary kernel estimation of $\mu(x^{+})-\mu(x^{-})$, where $\mu(x^{+})$ and $\mu(x^{-})$ are the right and left limits of $\mu(\cdot)$ at $x$. Thus, if $x$ is a continuous point of $\mu(\cdot)$, this quantity will be approximately zero at $x$. However, if $\mu(\cdot)$ is discontinuous at $x$, the quantity will be approximately equal to the jump size of $\mu(\cdot)$ at $x$. To establish our first main result, we need the following regularity conditions:

\begin{enumerate}[label=(\alph*),ref={Condition~(\alph*)}]
	\item \label{c1} There exist $0<\delta_2\le \delta_1<1$ such that $n^{-\delta_1}=\bigO(b)$ and $b=\bigO(n^{-\delta_2})$. 
	\item \label{c2} Assume $\EE|\epsilon_i|^p<\infty$ where $p>2/(1-\delta_1)$. 
	\item \label{c3} For the same $p$ defined in \ref{c2}, assume that $X_i\in \mathcal{L}^p$, $\theta_{n,p}=\bigO(\rho^n)$, and $\theta^*_{n,p}=\bigO(\rho^n)$ for some $0<\rho<1$.
	\item \label{c4} The density function $f$ of $X_i$ is positive on $[l-\epsilon,~u+\epsilon]$ for some $\epsilon>0$ and there exists a constant $B<\infty$ such that
	\[
	\sup_x \left[ |f_{X_n\,|\,\xi_{n-1}}(x)|+|f'_{X_n\,|\,\xi_{n-1}}(x)|+|f''_{X_n\,|\,\xi_{n-1}}(x)| \right]\le B,~\textrm{a.s}.
	\]
	\item \label{c5} $K(\cdot)$ is differentiable over $(0,1)$, the right derivative $K'(0+)$ and the left derivative $K'(1-)$ exists and  $\sup_{0\le u\le 1}|K'(u)|<\infty$. The Lebesgue measure of the set $\{u\in [0,1]: K(u)=0\}$ is zero. Furthermore, $K(0)=K(1)=0,~K'(0)>0$ and $\int_{0}^1uK(u)\dee u=0$. 
\end{enumerate}
For the above regularity conditions, \ref{c1} specifies the allowable range of the bandwidth. \ref{c2} puts a mild moment restriction on $\epsilon_i$. \ref{c3} requires that the quantities $\theta_{n,p}$ and $\theta_{n,p}^*$ satisfy the geometric moment contraction (GMC) property. The GMC property is preserved in many linear and nonlinear time series models such as the ARMA models and the ARCH and GARCH models; see \cite{Shao2007} for more discussions. Furthermore, denote $\Theta_n:=\sum_{i=0}^n\theta_{i,2}$, which measures the cumulative dependence of $X_0,...,X_n$ on $\eta_0$. Then if \ref{c3} holds, it is easy to see that  $\Theta_\infty<\infty$ which indicates short-range dependence of $\{X_i\}$. With \ref{c4}, we require that the density and conditional density of $X_i$ exist and are bounded. Moreover, $f$ has bounded derivatives up to the second order. \ref{c5} puts some restrictions on the smoothness and order of the kernel function $K$. In particular, $\int_{0}^1uK(u)\dee u=0$ indicates that $K$ is a second-order kernel which has both positive and negative parts on $[0,1]$.

\section{State-domain Change Point Detection and Estimation}\label{sec3}
In this section, we propose a test on the existence of change points in $\mu(\cdot)$ and an algorithm to estimate the number and locations of the change points when $\mu(\cdot)$ is discontinuous.

\subsection{Test for the existence of change points.}
With the foregoing discussion, we introduce a nonparametric statistic based on the density-weighted anti-symmetric kernel to test whether model \cref{new-model} has change points in the state domain regression function $\mu(\cdot)$ on $[l,u]$.
By proper scaling, our test statistic is defined as 
\[
t_n(x):=\frac{\sqrt{f(x)}}{\sigma(x)}\frac{1}{nb}\sum_{k=1}^n \tilde{K}_n\left(X_{k-1},x,b\right)X_k,
\]
where $\sigma^2(x)=\EE[\epsilon_i^2|X_{i-1}=x]$. In practice, since $f(\cdot)$ and $\sigma(\cdot)$ are unknown, we use the kernel density estimator $f_n(x)$ and Nadaraya--Watson (NW) estimator $\sigma^2_n(x)$ to replace $f(x)$ and $\sigma^2(x)$, respectively. The kernel density estimator is given by
\[
f_n(x)=\frac{1}{nh}\sum_{k=2}^n W\left(\frac{X_{k-1}-x}{h}\right),
\label{density}
\]
where $W(\cdot)$ is a general kernel function with $W(\cdot)\ge 0$ and $\int W(u)\dee u =1,~h=h_n$ is the bandwidth sequence satisfying $h\to 0$ and $nh\to \infty$. Let $\hat{e}_k^2=[X_k-\mu_n(X_{k-1})]^2$ be the square of the estimated residuals, where $$\mu_n(x)=\frac{1}{nh f_n(x)}\sum_{k=2}^n W\left(\frac{X_{k-1}-x}{h}\right)X_k$$ is the NW estimator of $\mu(\cdot)$, then the NW estimator of $\sigma^2(x)$ is given by
\[
\sigma^2_n(x)=\frac{1}{nh f_n(x)}\sum_{k=2}^n W\left(\frac{X_{k-1}-x}{h}\right)\hat{e}_k^2.\label{variance}
\]
The following remark provides the uniform consistency of the estimated density and conditional variance functions. 
\begin{remark}\label{remark1}
	Under \ref{c1} for both bandwidths $h$ and $b$ with $0<\delta_1<1/4$, \ref{c3}, \ref{c4}, and \ref{c5}, we have
	\[\label{eq_fn_bias}
	\EE f_n(x) -f(x)=f''(x)h^2\psi_W+o(h^2),
	\]
	where $\psi_W:=\int u^2 W(u)\dee u/2$
	and 
	\[
	\sup_x \left|f_n(x)-f(x)\right|=\bigO_{\Pr}\left(\frac{(\log n)^3}{\sqrt{nh}}+h^2\log n\right).
	\]
	Similarly, for $\sigma^2_n(x)$, under the conditions of
	\cref{thm_main}, we also have
	\[
	\sup_x \left|\sigma_n^2(x)-\sigma^2(x)\right|=\bigO_{\Pr}\left(\frac{(\log n)^3}{\sqrt{nh}}+h^2\log n\right).
	\]
See \cref{proof_remark1} for the proof.
\end{remark}

Let 
$f_{\epsilon}(\cdot)$ be the density function of $\epsilon_i$ and $\lambda_K=\int K^2(x)\dee x$. We have the following main result on the asymptotic properties of the proposed test statistic.

\begin{theorem}\label{thm_main}
	 Let $l,u\in\mathbb{R}$ be fixed. Recall the piece-wise formulation of \cref{e2}, let $T_j^{\epsilon}$ and $T^\epsilon$ be the $\epsilon$-neighborhoods of the intervals $T_j=[a_j,a_{j+1})$ and $T=[l,u]$, respectively. Let $T_a=\{a_j\}$ be the collection of the change points, $T_a^{\epsilon}$ be the $\epsilon$-neighborhood of $T_a$. Assume that \ref{c1}-\ref{c5} hold with $f_{\epsilon}(\cdot),~\sigma(\cdot) \in\mathcal{C}^3(T^{\epsilon}),~\mu_j(\cdot) \in \mathcal{C}^3(T_j^{\epsilon})$ for some $\epsilon>0$ and $b$ satisfies
	\[
	0<\delta_1<1/3,\quad 0<\delta_2\le 1/4,\quad nb^9\log n=o(1),
	\] 
	then
	\[
	\Pr\left(\sqrt{\frac{nb}{2\lambda_K}}\sup_{x\in T\cap (T_a^{b})^c} \left| t_n(x)\right| -d_n\le \frac{z}{(2\log \bar{b}^{-1})^{\frac{1}{2}}} \right)\to e^{-2e^{-z}},
	\] 
	where $\bar{b}:=b/(u-l)$ and
	\[
	d_n:=
	(2\log \bar{b}^{-1})^{\frac{1}{2}}+\frac{1}{(2\log \bar{b}^{-1})^{\frac{1}{2}}}\log\frac{\sqrt{K_2}}{\sqrt{2}\pi}
	\]
	with $K_2:=\int_{0}^{1} (K'(u))^2\dee u/\lambda_K$.
\end{theorem}
\begin{proof}
	See \cref{proof_thm_main}.
\end{proof}
\cref{thm_main} is a general result which establishes the asymptotic theory of the test statistic. In practical implementation, we will use the density estimates $f_n(x)$ and variance estimates $\sigma_n(x)$ instead of $f(x)$ and $\sigma(x)$ to calculate $t_n(x)$ as discussed before. Therefore, we have the following corollary.
 
\begin{corollary}\label{coro}	
	Denote $t_n^\ast(x)=\frac{\sqrt{f_n(x)}}{\sigma_n(x)}\frac{1}{nb}\sum_{k=1}^n \tilde{K}_n\left(X_{k-1},x,b\right)X_k$. Under the conditions of \cref{thm_main} and further assume the bandwidth $h\le b$, then the asymptotic result of \cref{thm_main} holds for $t_n^\ast(x)$; this is
	\[
	\Pr\left(\sqrt{\frac{nb}{2\lambda_K}}\sup_{x\in T\cap (T_a^{b})^c} |t_n^\ast(x)| -d_n\le \frac{z}{(2\log \bar{b}^{-1})^{\frac{1}{2}}} \right)\to e^{-2e^{-z}}.
	\] 	
\end{corollary}

Note that in \cref{coro}, we have added the assumption $h\le b$ with the purpose of ensuring the consistency of $f_n(x)$ and $\sigma_n(x)$ on $T\cap (T_a^{b})^c$. When there is no change point in $\mu(\cdot)$, we have similar results as shown in the following remark, which suggests that under the null hypothesis, after proper scaling and centering, our test statistic converges to a Gumbel distribution asymptotically.

\begin{remark}\label{remark2}
	Assume $H_0: M=0$ holds. We further assume that $f(\cdot),~\sigma(\cdot) \in \mathcal{C}^3(T^{\epsilon})$ and the remaining conditions of \cref{coro} hold. Then, $T_a=\emptyset$, $T_a^b=\emptyset$, which implies $T\cap (T_a^{b})^c=T$. Therefore, the previous theorem reduces to
	\[
	\Pr\left(\sqrt{\frac{nb}{2\lambda_K}}\sup_{x\in T} |t_n^\ast(x)| -d_n\le \frac{z}{(2\log \bar{b}^{-1})^{\frac{1}{2}}} \right)\to e^{-2e^{-z}}. \label{target}
	\] 
\end{remark}

Denote the jump-size of $\mu(\cdot)$ at $a_i$ as $\Delta_i$; that is, $\Delta_i:=|\mu(a_i+)-\mu(a_i-)|$. Next, we consider the alternative hypothesis $H_a: M\ge 1$ with $\Delta_i\ge \tilde{\Delta}>0$. When $H_a$ holds true, it is easy to see that the proposed test has an asymptotic power $1$ as $n\to\infty$. In other words, with some preassigned level $\alpha\in(0,1)$ and as $n\to\infty$, we have
\[
\Pr\left(\sup_{x\in T}|t_n(x)|\ge\sqrt{\frac{2\lambda_K}{nb}} \left[d_n-\frac{\log\{\log(1-\alpha)^{-1/2}\}}{(2\log\bar{b}^{-1})^{1/2}}\right]\right)\to 1.
\]
Once the null hypothesis of no change point is rejected, one would be interested in detecting the number of change points together with their locations, which we discuss in \cref{3.2}.

\subsection{Change-point Estimation}\label{3.2}
Suppose there exist a fixed number $M$ of change points on $[l,u]$, which are denoted by $l<a_1<\dots<a_M<u$, with the minimum jump size $\min_{1\le i\le M}\Delta_i\ge \tilde{\Delta}_n>0$. In this paper, we assume $\tilde{\Delta}_n=\bigO(1)$ which is allowed to decrease with $n$. 
The idea for estimating the number and locations of the change points is to search for local maximas of $|t_n(x)|$ which exceed the critical value of the test. To be more specific, we propose in the following a procedure for change point estimation.

\begin{itemize}
	\item For a fixed level $\alpha$, perform the bootstrap procedure described in \cref{4.1} to determine the critical value, say $C_{n,\alpha}$.
	\item Set  $T_1:=(l,u)$. 
	\item Starting from the interval $T_1$, find the largest $x$ of $|t_n(x)|$ that exceeds the critical value and denote its location as $\hat{a}_{(1)}$, then rule out the interval $[\hat{a}_{(1)}-b, \hat{a}_{(1)}+b]$ from $T_1$ to get $T_2:=T_1\cap [\hat{a}_{(1)}-b, \hat{a}_{(1)}+b]^c$.
	\item Repeat the previous step until all significant local maximas are found. In other words, $|t_n(x)|$ on the remaining intervals are all below $C_{n,\alpha}$.
	\item Denote the number of detected change points by $\hat{M}$ and re-order the estimated change points as $l<\hat{a}_1<\dots<\hat{a}_{\hat{M}}<u$. 
\end{itemize}

The following theorem provides an asymptotic result for $\hat{M}$ and $\hat{a}_i$.
	\begin{theorem}\label{thm_estimation}
		Under the conditions of \cref{thm_main}, we further assume that $K'(\cdot)$ is differentiable over $(0,1)$ with $K'(1)=0$, the right derivative $K''(0+)$ and the left derivative $K''(1-)$ exist and  $\sup_{0\le u\le 1}|K''(u)|<\infty$. The Lebesgue measure of the set $\{u\in [0,1]: K'(u)=0\}$ is zero. 
If $\sqrt{\frac{\log n}{nb}}=o(\tilde{\Delta}_n)$ then 
for any given level $\alpha$, we have
		\[
		\Pr\left(\left\{\hat{M}=M\right\}\cap \left\{\max_{1\le i\le M} |\hat{a}_i-a_i|< c_n \right\}\right)\to 1-\alpha,
		\]
		for any $c_n$ such that $1/c_n=\bigO\left(\tilde{\Delta}_n\sqrt{\frac{n}{b\log n}}\right)$
	\end{theorem}
	\begin{proof}
		See \cref{proof_thm_estimation}.
	\end{proof}
\cref{thm_estimation} reveals that for any given small probability $\alpha$, with asymptotic probability $1-\alpha$, our proposed procedure will correctly estimate all the change points with an accuracy $c_n$. It is important to mention that when $\tilde{\Delta}_n=\tilde{\Delta}>0$, that is, when the jump sizes have a fixed lower bound, 
the smallest order for $c_n$ is $\sqrt{b\log n/n}$, which is smaller than $n^{-1/2}$. It can also be seen as a product of $\sqrt{\log n}$ and the optimal convergence rate $(\sqrt{b/n})$ of time-domain change-point estimators established in \cite{Muller1992}. Hence, we conjecture that our rate $c_n$ is nearly optimal for state-domain change point detection. 
  
\section{Practical Implementation}\label{sec4}
\subsection{The bootstrap procedure}\label{4.1}
It is well known that the convergence rate of the Gumbel distribution in \cref{thm_main} is slow. As a result, a very large sample size would be needed for the approximation to be reasonably accurate. To overcome this issue, we propose a simulation-based bootstrap procedure to improve the finite-sample performance of the proposed test. The bootstrap procedure is as follows.

\begin{itemize}
	\item Generate i.i.d. standard normal random variables $U_k,~k=0,...,n$.
	\item Compute the quantity $\Pi_n^\ast$ defined in \cref{boots} for many times and calculate its $(1-\alpha)$th quantile as the critical value of our test.
\end{itemize}
For the proposed boostrap procedure, we have the following theoretical results which shows that, with proper scaling and centering, $\Pi_n^\ast$ has the same asymptotic Gumbel distribution.

 \begin{proposition} \label{prop}
	Denote $\Pi_n=\sup_{x \in T}|t_n^\ast(x)|$ and 
	\[
	\Pi_n^\ast=\sup_{x \in T}\left|\frac{\sqrt{g(x)}}{nb}\sum_{k=1}^n\tilde{K}_n(U_{k-1},x,b)U_k\right|,
	\label{boots}
	\]
	where $\{U_k\}_{k=0}^n$ are i.i.d. standard normal random variables and $g(x)$ is its density. Assume $H_0: M=0$, \ref{c1}, \ref{c5} hold and $b$ satisfies \[
	0<\delta_1<1/3,\quad 0<\delta_2\le 1/4,\quad nb^9\log n=o(1).
	\] 
	Then we have 
\begin{equation}\label{bootstrap}
\Pr\left(\sqrt{\frac{nb}{2\lambda_K}}\Pi_n^\ast-d_n\le \frac{z}{(2\log \bar{b}^{-1})^{\frac{1}{2}}} \right) \to {\rm e}^{-2{\rm e}^{-z}},~\text{as $n\to \infty$}.
\end{equation}
\end{proposition}


Proposition \ref{prop} shows that $\Pi_n^\ast$ and $\Pi_n$ have the same asymptotic Gumbel distribution with proper scaling and centering under the null hypothesis. Therefore, the $(1-\alpha)$th quantile of $\Pi_n$ can be estimated consistently by calculating the empirical $(1-\alpha)$th quantile $C_{n,\alpha}$ of $\Pi_n^\ast$ with a large number of replications by the bootstrap procedure. We reject the null hypothesis at level $\alpha\in(0,1)$ if $\Pi_n>C_{n,\alpha}$. When implementing the procedure described in \cref{3.2} for estimating the change points, we also suggest using $C_{n,\alpha}$ to find the detection region. Our numerical experiments suggest that the bootstrap method yields more accurate results than those based on the asymptotic limiting distribution under small or moderate sample sizes.

\subsection{Bandwidth selection}	
The bandwidth used in $f_n(x)$ can be chosen based on classic bandwidth selectors for nonparametric kernel density estimation. However, the choice of bandwidth $b$ for test statistic $t_n^\ast(x)$ and $h$ for the estimated variance  $\sigma_n^2(x)$ can be quite nontrivial and are usually of practical interest. In this paper, we adopt the standard leave-one-out cross-validation criterion for bandwidth selection suggested by \cite{Rice1991}: 
\[
{\rm CV}(b)&=\frac{1}{n}\sum_{k=1}^n\left[X_{k+1}-\mu_n^{(-k)}(X_k)\right]^2,\\
{\rm CV}(h)&=\frac{1}{n}\sum_{k=1}^n\left[(X_{k+1}-\mu_n(X_k))^2-\sigma_n^{2(-k)}(X_k)\right]^2
\]
where $\mu_n^{(-k)}(X_k)$ and  $\sigma_n^{2(-k)}(X_k)$ are the kernel estimators of $\mu$ and $\sigma^2$ computed with all measurements with the $k$th subject deleted, respectively. For example, a cross-validation bandwidth $\hat{b}$ can be obtained by minimizing ${\rm CV}(b)$ with respect to $b$, i.e., $\hat{b}=\mathop{\arg\min}_{b\in \mathcal{B}}{\rm CV}(b)$, where $\mathcal{B}$ is the allowable range of $b$. The bandwidth selection for $h$ is similar. 

\section{Simulation Study}\label{sec5}
In this section, we carry out Monte Carlo simulations to examine the finite-sample performances of our proposed test and estimator. Throughout the numerical experiments, the Epanechnikov kernel $W(x)=0.75(1-x^2)\mathbbm{1}(|x|\leq 1)$ is used for estimating the density and conditional variances. On the other hand, we adopt the higher-order kernel function in the form $K(x)=b[\tilde{W}(x)-a\tilde{W}(\sqrt{a}x)]$ in the expression of $\tK_n$, where $\tilde{W}(x)$ is the kernel function on [0,1] by shifting and scaling $W(x)$. From \cref{thm_main}, one can see that the power of our test increases as $\lambda_K$ decreases. As a result, we aim to maximize the quantity $Q(a,b)=\frac{\int_0^\infty K(x)\dee x}{\sqrt{\int_0^\infty K^2(x)\dee x}}$ with the constraints $\int_0^\infty K(x)\dee x=1$ and $\int_0^\infty xK(x)\dee x=0$ to choose $a$ and $b$. It turns out that $Q(a,b)$ is maximized at $a=0.34$ and $b=\frac{2}{\sqrt{0.34}-0.34}$. Hence, we will use $K(x)=\frac{2}{\sqrt{0.34}-0.34}[\tilde{W}(x)-0.34\tilde{W}
(\sqrt{0.34}x)]$ in our simulations and data analysis.

\subsection{Accuracy of bootstrap.}
We run Monte Carlo simulations to study the accuracy of the proposed bootstrap procedure for finite samples $n=200,~500$ and 800. Here, we aim to test the null hypothesis $H_0$ of no change point in the regression function. The number of replications is fixed to be 1000 and the number of bootstrap samples is $B=2000$ at each replication.

To guarantee the stationarity of the process $\{X_i\}$, $|\mu(x)|$ is required to be less than one \cite[Section 2.1]{Fan03}. First, we consider Model A listed below to investigate the robustness of our testing procedure with respect to various levels of persistence in the data generating process. Additional four state-domain nonlinear models (Models B - E listed below) where $\mu(\cdot)$ is of various shapes are further investigated for the accuracy of our test. In our simulations the martingale difference process $\epsilon_i=\sigma(X_{i-1})\epsilon_i^\ast$ with $\sigma^2(x)=\EE(\epsilon_i^2|X_{i-1}=x)$ and $\epsilon_i^\ast \stackrel{i.i.d.}{\sim}\mathcal{N}(0,1)$. Note that the error processes $\{\epsilon_i\}$ are specified via different conditional variance functions $\sigma^2(x)$ in Models A--D. On the other hand, in Model E we set {$\epsilon_i=0.5\eta_i(\eta_{i-7}+1.5)$} where $\eta_i\stackrel{i.i.d.}{\sim}\mathcal{N}(0,1)$ so that $\{\epsilon_i\}$ has a period of 7 which matches the data generating process observed in the empirical data example in Section \ref{sec6}. 

\begin{itemize}
	\item Model A: 
	\begin{align*}
		\mu(x)&=
		\begin{cases}
			\kappa_1 x^3, &|x|\leq 1,
			\cr \kappa_1, &x>1,
			\cr -\kappa_1, &x<-1,
		\end{cases}\\
		\sigma(x)&=1.5{\rm e}^{-0.5x^2},
	\end{align*}
where $\kappa_1=0.2, 0.4, 0.6, 0.8$ represents various levels of temporal dependence in the series.
	\item Model B:
	\begin{equation*}
		\mu(x)=0.2{\rm e}^{-0.5x^2},~\sigma(x)=\frac{1.5{\rm e}^{x}}{1+{\rm e}^{x}}.
	\end{equation*}
	\item Model C:
	\begin{align*}
		\mu(x)&=\frac{0.3{\rm e}^{x}}{1+{\rm e}^{x}}, \\
		\sigma(x)&=
		\begin{cases}
			0.7(1+x^2), &|x|\leq 1,
			\cr 1.4, &\text{otherwise}.
		\end{cases}
	\end{align*}
	\item Model D:
	\begin{equation*}
		\mu(x)=0.8\sin(x),~\sigma(x)=1.
	\end{equation*}
\item Model E:
\begin{equation*}
	\mu(x)=0.5\cos(x).
	\end{equation*}
\end{itemize}

Note that the regression functions $\mu(\cdot)$ in Models A--E are all continuous. At nominal significance levels $\alpha=0.05$ and $0.1$, the simulated Type I error rates for sample sizes $n=200, 500$ and $800$ are reported in Tables \ref{T1}--\ref{T2} for Model A and Models B--E, respectively. To measure the strength of the nonlinear temporal dependence, we will employ the auto-distance correlation function (ADCF) investigated in \cite{zhou2012measuring}. In Table \ref{T1}, we illustrate the first order ADCF (denoted by $\mathcal{R}(1)$) for Model A. Meanwhile, for Model E the first order and the seventh order ADCF are listed in Table \ref{T2}. One can see that the performance of our testing procedure is reasonably accurate for different sample sizes across the models and the accuracy improves as the sample size increases. On the other hand, from Table \ref{T1}, we find that as the dependence of the process becomes stronger, the type I errors tend to be less accurate, but are still in a reasonable range. 

\begin{table}[htbp!]
	\centering
	\caption{Simulated type I error rates for Model~A with the first order ADCF of the model.}
	\label{T1}
	\vspace{0.05in}
	\begin{tabular}{cccccccccc}
		\hline
		Model A&$\kappa_1$&0.2&0.4&0.6&0.8\\
		\hline
		&$\mathcal{R}(1)$&0.240&0.321&0.412&0.523\\
		\hline
		\multirow{3}{*}{$\alpha=0.05$}&$n=200$&0.046&0.055&0.064&0.067\\
		&$n=500$&0.048&0.052&0.060&0.065\\
		&$n=800$&0.053&0.049&0.050&0.065\\
		\hline
		\multirow{3}{*}{$\alpha=0.1$}&$n=200$&0.103&0.109&0.132&0.131\\
		&$n=500$&0.096&0.092&0.119&0.138\\
		&$n=800$&0.099&0.092&0.109&0.126\\
		\hline
	\end{tabular}
\end{table}

\begin{table}[htbp!]
	\centering
	\caption{Simulated type I error rates for Model~B--E with the first and seventh order ADCF of Model~E.}
	\label{T2}
	\vspace{0.05in}
	\begin{tabular}{cccccccccc}
		\hline
		&Model&B&C&D&\multicolumn{2}{c}{E}\\
		\hline
		\multirow{3}{*}{$\alpha=0.05$}&$n=200$&0.040&0.044&0.050
		&0.060&$\mathcal{R}(1)$\\
		&$n=500$&0.066&0.041&0.054
		&0.054&0.195\\
		&$n=800$&0.056&0.051&0.057&0.054\\
		\hline
		\multirow{3}{*}{$\alpha=0.1$}&$n=200$&0.085&0.106&
		0.124&0.105&$\mathcal{R}(7)$\\
		&$n=500$&0.102&0.092&0.114
		&0.092&0.258\\
		&$n=800$&0.093&0.101&0.112&0.095\\
		\hline
	\end{tabular}
\end{table}

\subsection{Power of hypothesis testing}
In this subsection, we consider the simulated power of our test under various alternatives. Recall the representation $\epsilon_i=\sigma(X_{i-1})\epsilon_i^\ast$ with $\epsilon_i^\ast \stackrel{i.i.d.}{\sim}\mathcal{N}(0,1)$. Here, we consider the
following two types of alternatives with a change point of size $\delta$ : 
\begin{itemize}
\item Model F$_1$:
\[\mu(x)&=
\begin{cases}
0.5{\rm e}^{-x^2}, &x<0,
\cr 0.5{\rm e}^{-x^2}-\delta, &x\geq0,
\end{cases}\\
\sigma(x)&={\rm e}^{-0.5x^2}.\]
\item Model F$_2$:
\[\mu(x)&=
\begin{cases}
0.3-\delta, &x<0,
\cr 0.3, &x\geq0,
\end{cases}\\
\sigma(x)&=\frac{{\rm e}^x}{1+{\rm e}^x}.\]
\end{itemize}

\begin{figure}[htbp!]
	\centering
	\includegraphics[scale=0.6]{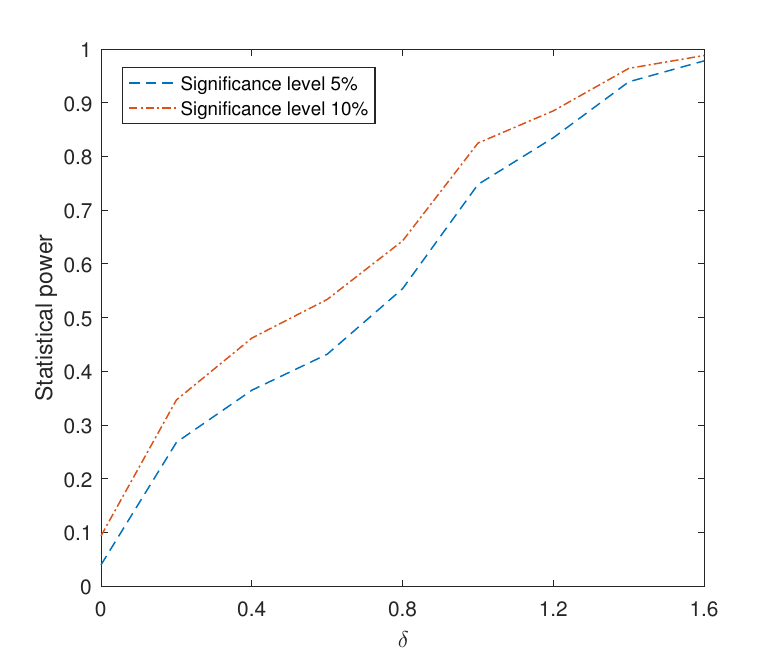}
	\caption{Simulated rejection rates for testing change point for Model F$_1$.}\label{F1}
\end{figure}

\begin{figure}[htbp!]
	\centering
	\includegraphics[scale=0.6]{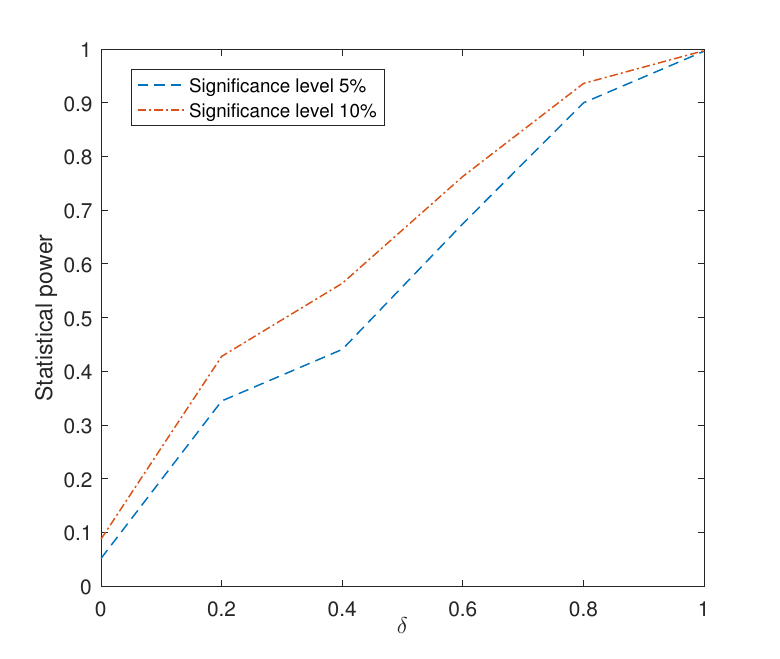}
	\caption{Simulated rejection rates for testing change point for Model F$_2$.}\label{F2}
\end{figure}

We let the jump size $\delta$ range from $0$ to $1.6$ for model F$_1$ and from $0$ to $1$ for model F$_2$ at location $x=0$. For each model, we investigate the empirical sensitivity of our testing procedure under nominal levels $0.05$ and $0.1$ with sample size $n=800$ based on $1000$ replications. The simulated power curves for the above models are plotted in \cref{F1} and \cref{F2}, respectively. According to the plots, the statistical power of the proposed testing procedure increases reasonably fast as $\delta$ increases. On the other hand, we also observe that our test shows a slower speed of increase at near alternatives when compared with ``classic" power curves of parametric tests. We believe that part of the reason is that our nonparametric test aims at detecting alternatives from a large class of discontinuous functions while tests tailored to some parametric models (such as the threshold model) target a specific class of alternative functions. Therefore our test is expected to be less sensitive to small deviations from the null compared to those parametric tests. See also Section \ref{sec:threshold} for a numerical experiment that compares the sensitivity of our testing procedure with that of a parametric test of the threshold model. 

\subsection{Accuracy in estimating the number of change points and their locations}

Utilizing the algorithm listed in \cref{3.2},  in this subsection we focus on estimating the number of change points and their  locations based on $1000$ realizations with sample sizes $n=200,~500$ and $800$. In the simulations, we let the error process $\{\epsilon_i^\ast\}_{i=1}^n$ be i.i.d.~standard normal random variables and consider the following two cases:

\begin{itemize}
	\item Case 1: A single change point.
	\[
	\mu(x)&=
	\begin{cases}
	0.7{\rm e}^{-x^2}, &x<0,
	\cr 0.7{\rm e}^{-x^2}-1.6, &x\geq0,
	\end{cases}\\
	\sigma(x)&={\rm e}^{-0.5x^2}.
	\]
	\item Case 2: Two change points.
	\[
	\mu(x)&=
	\begin{cases}
	0.8x+0.8, &x<-0.3,
	\cr -1, &-0.3\leq x<0,
	\cr -0.2x+0.5, &x\geq0,
	\end{cases}\\
	\sigma(x)&=\frac{{\rm e}^x}{1+{\rm e}^x}.
	\]
\end{itemize}

The estimates of the locations of change points are compared in terms of their mean absolute deviation errors (MADE) and mean squared errors (MSE). We also report the simulated percentage of correctly estimating the number of change points. The results are listed in \cref{T3}. One can see from \cref{T3} that the values of MADE and MSE are all quite small, which suggests the estimated locations by our approach are fairly accurate. Furthermore, as the sample size increases, the percentage of correctly estimating the number of change points increases in both cases.

\begin{table}[htbp!]
	\centering
	\caption{Accuracy in estimating the change-point locations and the percentage of correctly estimating the number of change points.}
	\label{T3}
	\begin{tabular}{lcccccccc}
		\hline
		Case 1&$n$&MADE&MSE&Percentage\\
		\hline
		\multirow{3}{*}{$a_1$=0}&200&0.0451&0.0055&90.51$\%$\\
		&500&0.0195&0.0014&93.77$\%$\\
		&800&0.0134&0.0006&94.51$\%$\\
		\hline
		Case 2&$n$&MADE&MSE&Percentage\\
		\hline
		$a_1=-0.3$&\multirow{2}{*}{200}&0.0519&0.0059
		&\multirow{2}{*}{$81.82\%$}\\
		$a_2=0$&&0.0757&0.0069&\\
		\hline
		$a_1=-0.3$&\multirow{2}{*}{500}&0.0508&0.0043&\multirow{2}{*}{86.59$\%$}\\
		$a_2=0$&&0.0496&0.0042\\
		\hline
		$a_1=-0.3$&\multirow{2}{*}{800}&0.0386&0.0028&\multirow{2}{*}{89.80$\%$}\\
		$a_2=0$&&0.0362&0.0024\\
		\hline
	\end{tabular}
	\vskip 1mm {\footnotesize\noindent Note: true change point  $a_1=0$ for Case 1; true change points $a_1=-0.3$ and $a_2=0$ for Case 2; MADE: mean absolute deviation error; MSE: mean squared error.}
\end{table}
 
\subsection{Comparison to threshold testing and estimation in threshold model}\label{sec:threshold}
In this subsection, we compare the accuracy and sensitivity of our nonparametric method with existing threshold testing and estimation methods for the classic threshold AR (TAR) model proposed by \cite{Tong1980} when the TAR model is indeed the underlying data generating mechanism. We consider the following two-regime TAR(1) model
 \begin{equation*}
 	X_{i}=
 	\begin{cases}
 		0.5(X_{i-1}+1)+\epsilon_i, &X_{i-1}<0.25,\\
 		\cr \kappa_2(X_{i-1}+1)+\epsilon_i, &X_{i-1}\ge 0.25,
 	\end{cases}
 \end{equation*}
where $\kappa_2=0.5, 0.3, 0.1, -0.1, -0.3, -0.5$ and the error process $\epsilon_i \stackrel{i.i.d.}{\sim}\mathcal{N}(0,0.75^2)$. First, we are interested in comparing the accuracy and power of our nonparametric test with the parametric $F$-test of threshold nonlinearity proposed in \cite{tsay1989testing}. Table \ref{T4} shows the testing results for nonlinearity of the model based on both the parametric and nonparametric methods, in which the sample size is $n=800$ and the number of bootstrap samples is $B=2000$.
 
 \begin{table}[htbp!]
 	\centering
 	\caption{Simulated rejection rates for testing change point with TAR(1) model.}
 	\label{T4}
 	\vspace{0.05in}
 	\begin{tabular}{cccccccccc}
 		\hline
 		$\kappa_2$&&0.5&0.3&0.1&$-0.1$&$-0.3$&$-0.5$\\
 		\hline
 		\multirow{2}{*}{Para.}&$\alpha=0.05$&0.042&0.175&0.831&0.904&1&1\\
 		&$\alpha=0.1$&0.095&0.282&0.897&0.906&1&1\\
 		\hline
 		\multirow{2}{*}{Nonpara.}&$\alpha=0.05$&0.069&0.256&0.406&0.646&0.792&0.910\\
 		&$\alpha=0.1$&0.131&0.378&0.540&0.761&0.861&0.940\\
 		\hline
 	\end{tabular}
 \end{table}
 
 We observe that the nonparametric method has slightly higher powers when the scale coefficient $\kappa_2$ changes slightly from 0.5. However, as $\kappa_2$ becomes $0.1$ or smaller, the parametric method has higher powers than the nonparametric method.
 
In addition, we compare the accuracy in change point estimation. We study the following TAR(1) model,
 \begin{equation*}
 	X_i=
 	\begin{cases}
 		\frac{2}{3}(X_{i-1}+1)+\epsilon_i, &X_{i-1}<0.25,\\
 		\cr -\frac{2}{3}(X_{i-1}+1)+\epsilon_i, &X_{i-1}\ge 0.25,
 	\end{cases}
 \end{equation*}
 where $\epsilon_i \stackrel{i.i.d.}{\sim}\mathcal{N}(0,0.75^2)$.
 Note that parametric estimation of the threshold value of the above two-regime TAR(1) process can be done via the R function \textsf{uTAR} in the \textsf{NTS} package (we refer to \cite{LCT20} for more details). 
 The simulated MADEs and MSEs are listed in \cref{T5}. From \cref{T5}, one can see that both methods provide relatively accurate estimates of the locations of change point (threshold). The parametric method shows more accurate estimation results comparing with those of the nonparametric method. With the above observations, it can be seen that the parametric method is better for testing and detecting change points for the TAR model when the model is well-specified. This result is not surprising since testing sensitivity and estimation accuracy tend to be higher when the model is correctly restricted to a (smaller) parametric class. 
 
 \begin{table}[htbp!]
 	\centering
 	\caption{Estimation accuracy for change-point locations.}
 	\label{T5}
 	\vspace{0.05in}
 	\begin{tabular}{lcccccccc}
 		\hline
 		&$n$&MADE&MSE\\
 		\hline
 		\multirow{3}{*}{Nonpara.}&200&0.1055&0.0143\\
 		&500&0.0519&0.0066\\
 		&800&0.0367&0.0041\\
 		\hline
 		\multirow{3}{*}{Para.}&200&0.0340&0.0027\\
 		&500&0.0178&0.0012\\
 		&800&0.0098&0.0004\\
 		\hline
 	\end{tabular}
 \end{table}

\section{Illustrative example}\label{sec6}

In this section, we consider the daily new confirmed cases of Coronavirus disease of 2019 (COVID-19) in Germany. The dataset contains $156$ observations from April 28th to September 30th of 2020 which can be downloaded from  \href{url}{https://ourworldindata.org/coronavirus-source-data}. From the COVID-19 timeline, Germany registered the first case on January 28th and later suffered an outbreak of this pandemic from mid March to late April. In the data analysis, we select the aforementioned time span between the first and second waves of COVID-19 so that the time series is approximately stationary. Let $X_i$ be the logarithm of confirmed cases at day $i=1,...,156$ and $Y_i=X_{i+1}-X_i$. The sample path of $\{X_i\}$ and the ADCF plot of $\{X_i\}$ are shown in \cref{F3}. Both plots in \cref{F3} suggest that the time series is approximately stationary and has a moderate seasonal dependence with period $S=7$. The seasonal behaviour probably comes from the reporting lag behind during weekends, which happens in almost every country. We consider the state-domain nonlinear regression model (which is equivalent to \cref{new-model}): 
\begin{equation}\label{e.2}
Y_i = \mu(X_i)+\epsilon_i,
\end{equation}
where $\{\epsilon_i\}$ is a martingale difference sequence. In this application, $\mu(x)$ represents the expected increase or decrease in percentage of COVID-19 cases in day $i$ when $X_{i-1}=x$.

\begin{figure}[t]
	\flushleft
	\subfigbottomskip=-5pt
	\subfigure{
		\begin{minipage}{10cm}
			\flushleft
			\includegraphics[trim=60 0 60 0,height=2.5cm,width=13cm]{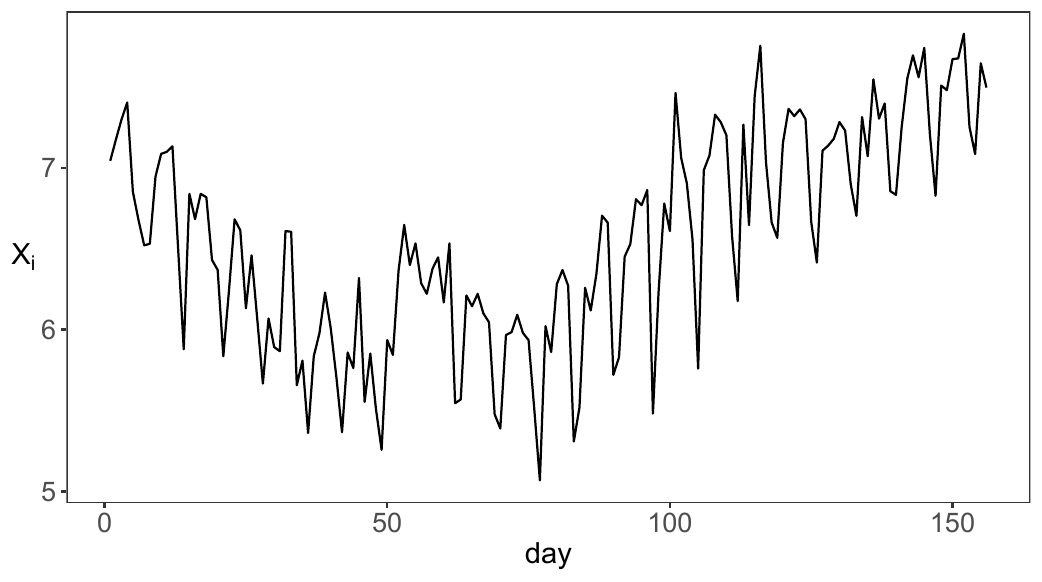}
	\end{minipage}}
	\subfigure{
		\begin{minipage}{10cm}
			\flushleft
			\centering
			\includegraphics[trim=60 0 60 0,height=4.5cm,width=13.8cm]{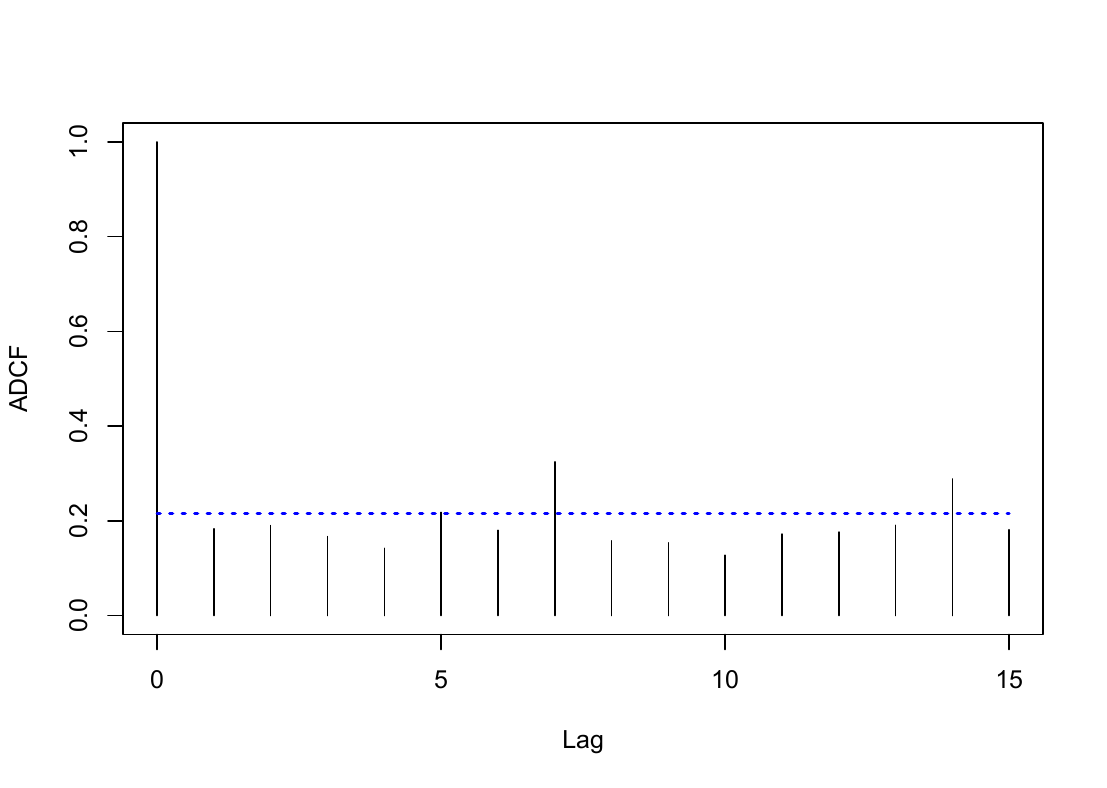}
	\end{minipage}}
	\caption{Top: Logarithm of daily confirmed cases of COVID-19 in Germany from April 28th to September 30th, 2020 (which is denoted as $\{X_i\}_{i=1}^{156}$ in Section \ref{sec6}). Bottom: ADCF plot of $\{X_i\}$.}
	\label{F3}
\end{figure}


We apply the proposed method to testing whether $\mu(\cdot)$ contains any change points. 
We choose $T=[l,u]=[5.7,7.5]$ which includes 82.69$\%$ of $X_i$ so that data are relatively abundant in this region and the test is expected to be accurate. According to the leave-one-out cross-validation criterion, the selected bandwidths $b$ and $h$ are $0.446$ and $0.40$, respectively. Through the practical implementation in \cref{4.1}, we calculate the empirical 99$\%$ quantile of $\Pi_n^\ast$ with 10000 bootstrap samples, which gives $C_{n,\alpha}=1.596$. Next, we investigate the behaviour of the test statistics, which is shown in \cref{F4}. Our test rejects the null hypothesis of continuity of $\mu(\cdot)$ at $1\%$ level and flags two change points at $\hat{x}_1=6.83$ and $\hat{x}_2=7.40$.
\begin{figure}[htbp!]
	\centering
	\includegraphics[height=5.5cm,width=10cm]{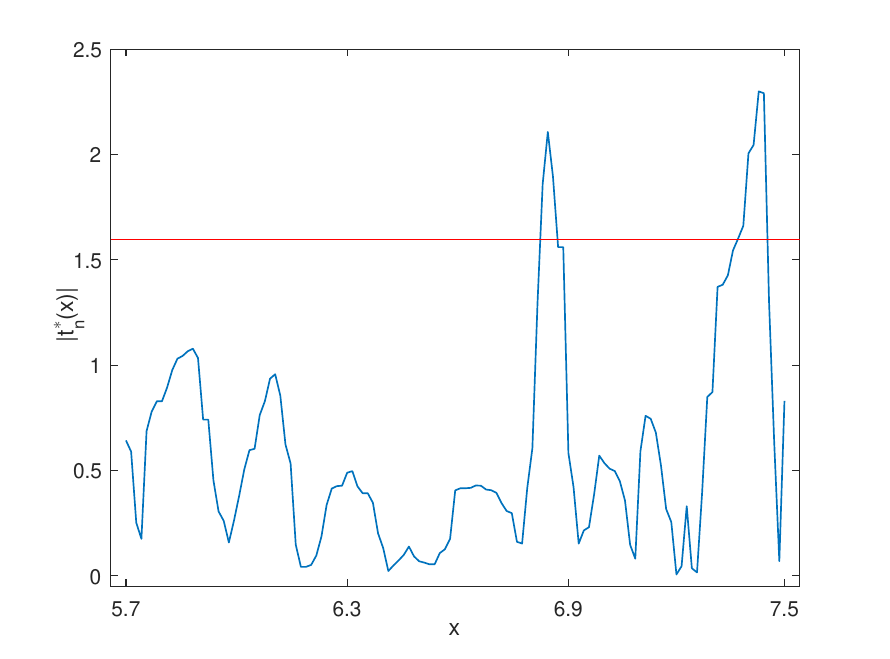}
	\caption{The absolute value of test statistics $|t_n^\ast(x)|$ over $[5.7,7.5]$, red line denotes the $99\%$ sample quantile (=1.596) of $\Pi_n^\ast$.}
	\label{F4}
\end{figure}

Note that $Y_i$ can be viewed as the conditional daily growth rate for COVID-19. For comparison, we also use the nonparametric local polynomial method to fit $\mu(x)$ assuming that there is no change point. The corresponding estimated regression function $\mu_n(x)$ over $[5.7,7.5]$ is plotted on the left hand side of \cref{F5}. On the right hand side of \cref{F5} we plot the fitted drift function $\mu_n(x)$ with the knowledge of the change points. The difference between the two plots in \cref{F5}  suggest that, with or without the knowledge of change points, our understanding of the relationship between $Y_i$ and $X_i$ can be quite different. 
With the knowledge of the change points, we can see that two large jumps exist at $x_1=6.83$ and $x_2=7.40$, which shows that the growth rate changes abruptly at these two points.

\begin{figure}[htbp!]
	\centering
	\subfigure{
		\begin{minipage}{6.5cm}
			\centering
			\includegraphics[height=4.5cm,width=6.5cm]{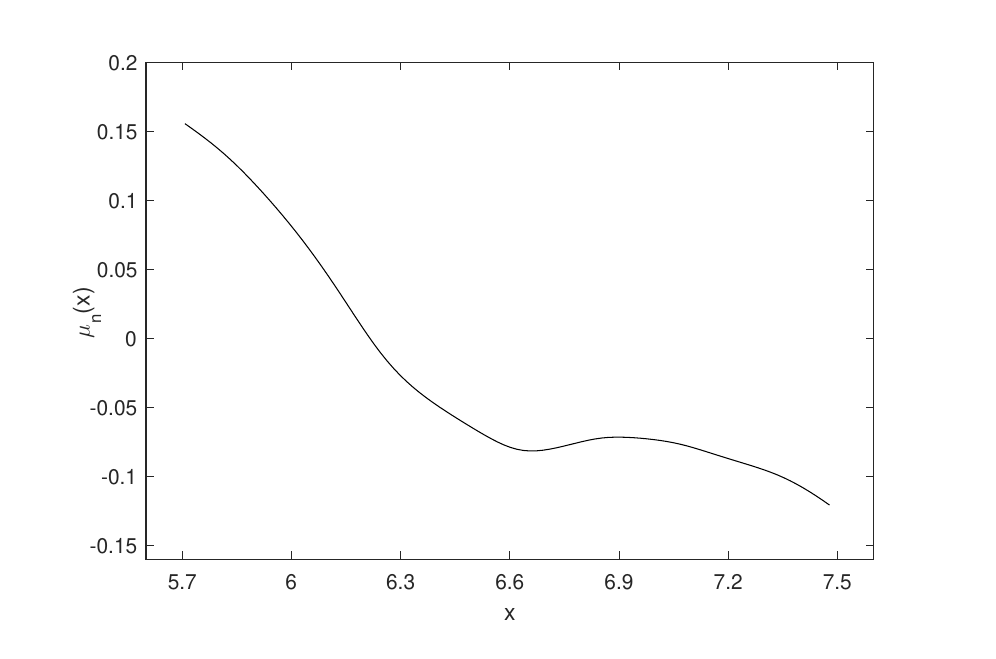}
	\end{minipage}}
	\hspace{-10mm}
	\subfigure{
		\begin{minipage}{6.5cm}
			\centering
			\includegraphics[height=4.5cm,width=6.5cm]{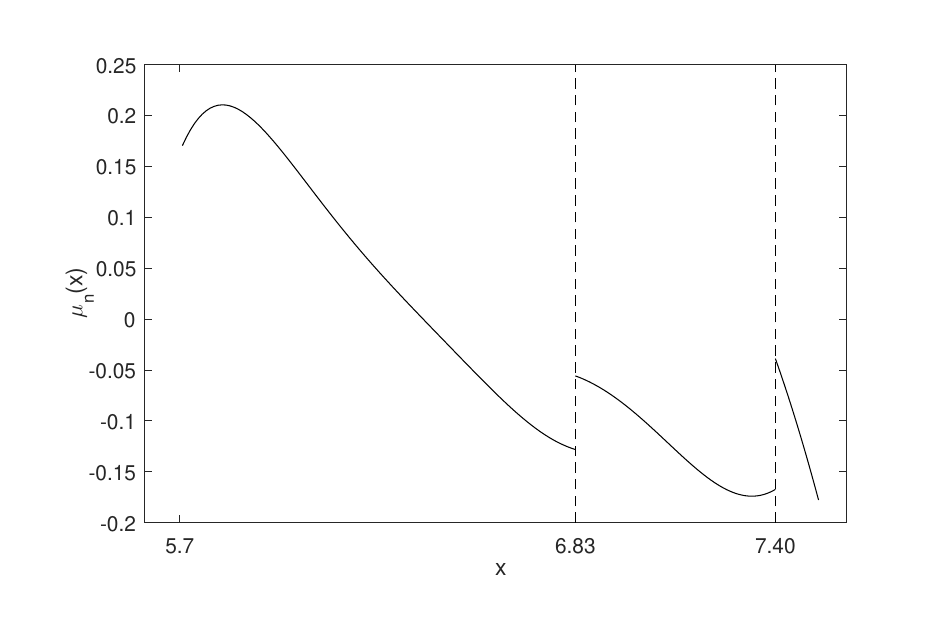}
	\end{minipage}}
	\caption{Left: Smooth fit with no change points;
		Right: Piece-wise smooth fit with the knowledge of two change points.}
	\label{F5}
\end{figure}

It is obvious to see from the right plot of \cref{F5} that those two change points divide the state domain into three regimes/phases. Furthermore, the latter plot indicates that the nonlinear dynamics can be approximated by a three-regime threshold model with the data generating mechanism switching at the detected change points.
Additionally, according the timeline, we can find out the periods corresponding to each phase. The first phase $x\in[5.7,6.83)$ contains May 3-5, 10-11, May 15-August 5, August 9-10, 16-17, 23-24 and 30-31 where the trajectory depicts a relatively inactive period of the virus transmission and the conditional infection rate $\mu(x)$ decreases from positive to negative as $x$ increases. The second phase $x\in[6.83,7.4)$ includes April 28-30, May 6-9, August 7-9, 11-15, 25-29, September 1-6, 8-9, 13-15 and 20-21 where the conditional infection rate jumps up when $x$ surpasses $6.83$ and then it decreases gradually again. 
The third phase $x\in [7.4,7.83]$ corresponds to August 20-21, September 16-19 and 22-26 where a sudden large increase in the conditional infection rate can be found at the left boundary and then it decreases sharply, possibly due to strong governmental interventions. 


In summary, the analyzed period from April 28th to September 30th of 2020 of German COVID-19 data shows a complicated nonlinear dynamic balance between disease transmission and government intervention. The proposed method of the paper could help understand this complex nonlinear dynamics by determining the boundaries of phases where the state-domain relationship changes abruptly and subsequently segment the time series into multiple regimes.

\vspace{0.5cm}

{\noindent\bf Acknowledgments.} The authors are grateful to the editor, Prof.~Serena Ng, and two anonymous referees for their valuable comments and suggestions which significantly improved the quality of the paper. Zhou Zhou's research has been partially sponsored by NSERC (No.489079). Yan Cui is supported by the China Scholarship Council (No.201806170148).

\section{Proofs of main results}\label{sec7}
\subsection{Proof of \cref{thm_main}}\label{proof_thm_main}

The outline of the proof is as follows. Firstly, we use the following decomposition of $X_i$
\[
X_i&=\mu(X_{i-1})+\epsilon_i =\left[\mu(X_{i-1})-\mu(x)\right] + \mu(x) + \epsilon_i,
\]
and prove the results involving the first two terms. This is given in \cref{proof_step1}.

Secondly, we use a technique called $m$-dependent approximation to approximate the martingale $\{\epsilon_i\}$ using $\{\EE[\epsilon_k\mid \xi_{i,i-m}]-\EE[\epsilon_k\mid \xi_{i-1,i-m}]\}$, where $\xi_{k_1,k_2}:=(\eta_{k_1},\dots,\eta_{k_2})$, for a properly chosen order $m\rightarrow\infty$, which simplifies the sum of a sequence of dependent random variables to a corresponding sum of $m$-dependent random variables. This is done in \cref{proof_step2}.

Thirdly, we divide the sequence of $n$ ($m$-dependent) random variables into alternating big and small blocks, where the length of big blocks has a slightly higher order than that of the small blocks. Furthermore, the length of the small blocks is larger than $m$. Using this proof technique, we can approximate the sum of $n$ ($m$-dependent) random variables using the sum of the subsequence which includes the random variables residing in the big blocks. Since the length of small blocks is larger than $m$, the $m$-dependent random variables in different big blocks are \emph{independent}. This part of the proof is given in \cref{proof_step3}.

Fourthly, we only need to deal with a sequence of independent sums of random variables within each big block. In order to get prepared for using the multivariate Gaussian approximation result by \cite{Zaitsev1987}, we first compute the asymptotic covariance structure of the sequence of independent sums. This is given in \cref{proof_step4}.

In the final two steps, we first apply the multivariate Gaussian approximation by \cite{Zaitsev1987}, which is given in \cref{proof_step5} and then prove the convergence to Gumbel distribution, which is given in \cref{proof_step6}. The techniques used in these two steps heavily depend on some existing work, particularly, the work by \cite{Zhao2008,Liu2010a}, which eventually applied the work by \cite{Bickel1973,Rosenblatt1976}.

\subsubsection{Decomposition}\label{proof_step1}
First, we {substitute $X_i=\mu(X_{i-1})+\epsilon_i$} into $t_n(x)$ and separate the terms involving $K$ and $K^*$. We first focus on the term involving $K$ only. That is,

\[\label{eq_decomposition}
\begin{split}
&\frac{1}{nb w(x,b)}\sum_{k=1}^n K\left(\frac{X_{k-1}-x}{b}\right)
\left[ \mu(X_{k-1})+{\epsilon_k}\right],\\
=&\frac{1}{nb w(x,b)}\sum_{k=1}^n K\left(\frac{X_{k-1}-x}{b}\right)
\left[ \mu(X_{k-1})-\mu(x)\right]\\
&+ \frac{1}{nb w(x,b)}\sum_{k=1}^n K\left(\frac{X_{k-1}-x}{b}\right)\mu(x)\\
&+\frac{1}{nb w(x,b)}\sum_{k=1}^n K\left(\frac{X_{k-1}-x}{b}\right){\epsilon_k}.
\end{split}
\]
Next it is easy to see that by the definition of $w(x,b)$, the second term of the decomposition on the right hand side of \cref{eq_decomposition} equals $\mu(x)$. 
For the first term of the decomposition in \cref{eq_decomposition}, following exactly the proof of \citep[Lemma 5.2]{Liu2010a}, uniformly over $x$, we have that

\[
\begin{split}
&\frac{1}{nb w(x,b)}\sum_{k=1}^n K\left(\frac{X_{k-1}-x}{b}\right)
\left[ \mu(X_{k-1})-\mu(x)\right]\\
=&\frac{b^2\psi_{K} \left[\mu''(x)f(x)+2\mu'(x)f'(x) \right]}{\EE[w(x,b)]+\bigO_{\Pr}(\sqrt{\log n/nb})}+\bigO_{\Pr}(b^3)+\bigO_{\Pr}(\tau_n)\\
=&\frac{b^2\psi_{K} \left[\mu''(x)f(x)+2\mu'(x)f'(x) \right]}{\EE[w(x,b)]}+b^2\bigO_{\Pr}(\sqrt{\log n/nb})\\
&+\bigO_{\Pr}\left(\sqrt{\frac{b\log n}{n}}+b^3+\frac{b}{n}\sqrt{\sum_{k=-n}^\infty (\Theta_{n+k}-\Theta_k)^2}\right)\\
=&\frac{b^2\psi_{K} \left[\mu''(x)f(x)+2\mu'(x)f'(x) \right]}{\EE[w(x,b)]}+\bigO_{\Pr}\left(\sqrt{\frac{b\log n}{n}}+b^3\right),
\end{split}
\]
where $\tau_n:=\sqrt{\frac{b\log n}{n}}+b^4+\frac{b}{n}\sqrt{\sum_{k=-n}^\infty (\Theta_{n+k}-\Theta_k)^2}$ comes from \citep[Lemma 2(ii)]{Zhao2008}, and in the last equality we have applied the assumptions on $b$ and $\sum_{k=-n}^\infty (\Theta_{n+k}-\Theta_k)^2$ to get $\frac{b}{n}\sqrt{\sum_{k=-n}^\infty (\Theta_{n+k}-\Theta_k)^2}=\bigO(\sqrt{b\log n/n})$.

\subsubsection{$m$-dependent approximation}\label{proof_step2}

For the third term of the decomposition in \cref{eq_decomposition}, recalling that we have defined the notation $\xi_{k_1,k_2}:=(\eta_{k_1},\dots,\eta_{k_2})$, 
we consider the decomposition of $\epsilon_k$,

	\[
	\epsilon_k=&\left(\epsilon_k-\EE[\epsilon_k\mid \xi_{k,k-m}]\right)\\
	&\quad+\left(\EE[\epsilon_k\mid \xi_{k,k-m}]-\EE[\epsilon_k\mid \xi_{k-1,k-m}]\right)\\
	&\quad+\EE[\epsilon_k\mid \xi_{k-1,k-m}],
	\]
	where $m=\lfloor n^{\tau}\rfloor$ where $\tau<1-\delta_1$. The first and last terms in the decomposition can be ignored comparing to the second term. To see this, consider
	\[
	\EE[\epsilon_k\mid \xi_{k-1,k-m}]&=\EE[\epsilon_k\mid \xi_{k-1,k-m}]-\EE[\epsilon_k\mid \mathcal{F}_{k-1}]\\
	&=\sum_{i=1}^{\infty}\EE[\epsilon_k\mid \xi_{k-1,k-i}]-\EE[\epsilon_k\mid \xi_{k-1,k-i-1}],
	\]
	which implies
	$\|\EE[\epsilon_k\mid \xi_{k-1,k-m}]\|_p=\bigO\left(\sum_{i=m}^{\infty}\rho^i\right)=\bigO(\rho^m)$. Since $m>(\log n)^2$, we have
	
	\[
	\sqrt{nb}\sup_{x \in T}\left|\frac{1}{nb}\sum_{k=1}^n K\left(\frac{X_{k-1}-x}{b}\right)\EE[\epsilon_k\mid\xi_{k-1,k-m}]\right|=\sqrt{\frac{n}{b}}\bigO_{\Pr}(\rho^m)=o_{\Pr}\left((\log n)^{-2}\right).
	\]
	Similarly, one can verify in the same way that
	\[
	\sqrt{nb}\sup_{x \in T}\left|\frac{1}{nb}\sum_{k=1}^n K\left(\frac{X_{k-1}-x}{b}\right)\left(\epsilon_k-\EE[\epsilon_k\mid\xi_{k,k-m}]\right)\right|=o_{\Pr}\left((\log n)^{-2}\right).
	\]
	Furthermore, since the martingale differences are uncorrelated, we have
	\[
	\EE[\epsilon_k^2]-\EE\left[\left(\EE[\epsilon_k\mid \xi_{k,k-m}]-\EE[\epsilon_k\mid \xi_{k-1,k-m}]\right)^2\right]=\bigO(\rho^m).
	\]
	Therefore, defining 
		\[
		\zeta_k:=\frac{\EE[\epsilon_k\mid \xi_{k,k-m}]-\EE[\epsilon_k\mid \xi_{k-1,k-m}]}{\sqrt{\EE\left[\left(\EE[\epsilon_k\mid \xi_{k,k-m}]-\EE[\epsilon_k\mid \xi_{k-1,k-m}]\right)^2\right]}}
		\]
	we have
	\[
	\sqrt{nb}\sup_{x\in T}\left|
	\frac{1}{nb}\sum_{k=1}^n K\left(\frac{X_{k-1}-x}{b}\right)\left({\zeta_k}-\frac{\epsilon_k}{\sqrt{\EE[\epsilon_k^2]}}\right) \right|=o_{\Pr}\left((\log n)^{-2}\right).
	\]

Next, following exactly the proof of \citep[Lemma 5.3]{Liu2010a}, we get that uniformly over $x$
\begin{align*}
&\frac{1}{nb w(x,b)}\sum_{k=1}^n K\left(\frac{X_{k-1}-x}{b}\right){\epsilon_k}\\
=&{\frac{1}{nb w(x,b)}\sum_{k=1}^n K\left(\frac{X_{k-1}-x}{b}\right){\sigma(X_k)\zeta_k}+\bigO_{\Pr}\left(\sqrt{\frac{b\log n}{n}}\right)}\\
=&\frac{1}{nb}\frac{1}{\EE[w(x,b)]+\bigO_{\Pr}(\sqrt{\log n/nb})}\sum_{k=1}^n K\left(\frac{X_{k-1}-x}{b}\right)\sigma(x){\zeta_k}+\bigO_{\Pr}\left(\sqrt{\frac{b\log n}{n}}\right)\\
=&\frac{1}{nb}\frac{1}{f(x)+\bigO_{\Pr}(b^2+\sqrt{\log n/nb})}\sum_{k=1}^n K\left(\frac{X_{k-1}-x}{b}\right)\sigma(x){\zeta_k}+\bigO_{\Pr}\left(\sqrt{\frac{b\log n}{n}}\right).
\end{align*}
Following the above arguments again we can compute the orders for the decomposition of the term involving $K^*$ and get $t_n(x)$ by the differences. Note that many terms such as $\mu(x)$ in the second term and $\bigO(b^2)$ term in the first term cancel out. 
Therefore, overall it can be easily verified that
\[
\begin{split}
t_n(x)&= \frac{\sqrt{f(x)}}{\sigma(x)}\frac{1}{nb f(x)}\sum_{k=1}^n\tK\left(\frac{X_{k-1}-x}{b}\right)\sigma(x){\zeta_k}+\bigO_{\Pr}\left(\sqrt{\frac{b\log n}{n}}+b^3\right)\\
&\quad + \bigO_{\Pr}(b^2+\sqrt{\log n/nb})\bigO_{\Pr}(\sqrt{\log n}),
\end{split}
\]
where $\tK(\cdot)$ is an anti-symmetric kernel defined by
\[
\tilde{K}(u):=K(u)-K^*(u).
\] 
Now to prove \cref{thm_main}, it suffices to show
\[
\Pr\left(\sqrt{\frac{nb}{2\lambda_K}}\sup_{x\in T}\frac{1}{\sqrt{f(x)}}\left|M_n(x)-M_n^*(x)\right|-d_n\le \frac{z}{(2\log \bar{b}^{-1})^{1/2}}\right)\to e^{-2e^{-z}},
\]
where $$M_n(x):=\frac{1}{nb}\sum_{k=1}^nK\left(\frac{X_{k-1}-x}{b}\right){\zeta_k},~M_n^*(x):=\frac{1}{nb}\sum_{k=1}^nK^*\left(\frac{X_{k-1}-x}{b}\right){\zeta_k}.$$

Note that we {have $\EE[\zeta_i]=0$ and $\EE[\zeta_i^2]=1$.}  Next, we define a truncated version of {$\zeta_i$}  by
\[
\breve{\zeta}_i:={\zeta_i\mathbbm{1}\{|\zeta_i|\le (\log n)^{12/(p-2)}\}-\EE\left[\zeta_i\mathbbm{1}\{|\zeta_i|\le (\log n)^{12/(p-2)}\}\right].}
\]

We next {define $\tilde{M}_n(x)$ using $m$-dependent conditional} expectations 
\[
\tilde{M}_n(x):=\frac{1}{nb}\sum_{k=1}^n \frac{\breve{\zeta}_k}{\breve{\sigma}^2} \left\{\EE\left[K\left(\frac{X_{k-1}-x}{b}\right)\,|\, \xi_{k-1,k-m}\right]\right.\\
\quad -\left.\EE\left[K\left(\frac{X_{k-1}-x}{b}\right)\,|\, \xi_{k-2,k-m}\right] \right\},
\]
where $\breve{\sigma}^2:=\EE \breve{\zeta}_1^2$.

\subsubsection{Alternating big and small blocks}\label{proof_step3}

{
	Recall that $m=\lfloor n^{\tau}\rfloor$. We choose $\tau_1$ such that $\tau<\tau_1<1-\delta_1$ and split $[1,n]$ into alternating big and small blocks $H_1,I_1,\cdots,H_{\iota_n},I_{\iota_n}, I_{\iota_n+1}$ with length $|H_i|=\lfloor n^{\tau_1}\rfloor$, $|I_i|=\lfloor n^{\tau}\rfloor$, $\forall 1\le i\le \iota_n$, and $|I_{\iota_n+1}|=n-\iota_n(\lfloor n^{\tau_1}\rfloor+\lfloor n^{\tau}\rfloor)$. Note that $\iota_n=\lfloor n/(\lfloor n^{\tau_1}\rfloor+\lfloor n^{\tau}\rfloor)\rfloor$.
	Then we define
	
	\[
	u_j(x):&= \sum_{k\in H_j}\frac{\breve{\zeta}_k}{\breve{\sigma}^2} \left\{\EE\left[{K}\left(\frac{X_{k-1}-x}{b}\right)\,|\, \xi_{k-1,k-m}\right]\right.\\
	&\quad -\left.\EE\left[{K}\left(\frac{X_{k-1}-x}{b}\right)\,|\, \xi_{k-2,k-m}\right] \right\}.
	\]
	Then we define
	\[
	\widetilde{M}_n(x):=
	\frac{1}{nb}\sum_{j\in \cup_{i=1}^{\iota_n} H_i} u_j(x).
	\]
}

Next we show in the following that {we can approximate $M_n(x)$ by $\tilde{M}_n(x)$ and then approximate $\tilde{M}_n(x)$ by $\widetilde{M}_n(x)$. That is, we show 
	\[\label{lemma_approximate_Mn}
\Pr\left(\sqrt{nb}\sup_{x\in T}\left| {M}_n(x) -\widetilde{M}_n(x) \right|\ge (\log n)^{-2}\right)=o(1).
\]
}

To show \cref{{lemma_approximate_Mn}}, we first follow the proof of \citep[Lemma 5.1]{Liu2010a} using Freedman's inequality for martingale differences \cite{Freedman1975} to get
\[
\Pr\left(\sqrt{nb}\sup_{x\in T}\left|
\frac{1}{nb}\sum_{k=1}^n K\left(\frac{X_{k-1}-x}{b}\right)({\zeta_k}-\breve{\zeta}_k) \right|\ge 3(\log n)^{-2}\right)=o(1),
\]
which implies we can approximate $M_n(x)$ by replacing {$\zeta_k$} with $\breve{\zeta}_k$ in the definition of $M_n(x)$.

Next, we write $K\left(\frac{X_{k-1}-x}{b}\right)$ as a sum of three terms

\[\label{eq_decompose_K}
\begin{split}
	&K\left(\frac{X_{k-1}-x}{b}\right)\\
	=&\left\{K\left(\frac{X_{k-1}-x}{b}\right)- \EE\left[K\left(\frac{X_{k-1}-x}{b}\right)\,|\, \xi_{k-1,k-m}\right] \right\}\\
	&\quad +\left\{\EE\left[K\left(\frac{X_{k-1}-x}{b}\right)\,|\, \xi_{k-1,k-m}\right]- \EE\left[K\left(\frac{X_{k-1}-x}{b}\right)\,|\, \xi_{k-2,k-m}\right] \right\}\\
	&\quad + \EE\left[K\left(\frac{X_{k-1}-x}{b}\right)\,|\, \xi_{k-2,k-m}\right].
\end{split}
\]
{Note that $\breve{\zeta}_k$ is uncorrelated with the second term of the right hand side of \cref{eq_decompose_K}.}
Next, we show that under our assumptions on physical dependence measure, the first term
of the right hand side of \cref{eq_decompose_K} 
becomes very small for large $m$. In order to rigorously prove this fact,
defining

\[
Z_k(x)=\breve{\zeta}_k\left\{K\left(\frac{X_{k-1}-x}{b}\right)- \EE\left[K\left(\frac{X_{k-1}-x}{b}\right)\,|\, \xi_{k-1,k-m}\right] \right\},
\]
we first approximate $\sum_{k=1}^n Z_k(x)$ by the skeleton process $\sum_{k=1}^n Z_k(x_j), 1\le j\le q_n$, where $q_n=\lfloor n^2/b\rfloor$ and $x_j=j/(bq_n)$. {Following the same arguments as in \citep[Proof of Lemma 4.2]{Liu2010a} using Freedman's inequality for martingale differences \cite{Freedman1975}}, we have 
\[
\sup_{x_{j-1}\le x\le x_j} \left| \sum_{k=1}^n (Z_k(x)-Z_k(x_j))\right| =o_{\Pr}\left(\sqrt{nb}/(\log b^{-1})^2\right).
\]
Next, we show $\sup_{x\in T} \EE|Z_k(x)|$ exponentially decays with $m$. We consider two cases $|X_{k-1}-\EE(X_{k-1}\,|\,\xi_{k-1,k-m})|\ge \rho_1^m$ and $|X_{k-1}-\EE(X_{k-1}\,|\,\xi_{k-1,k-m})|< \rho_1^m$, where $\rho_1=\frac{1+\rho}{2}$. Using the assumption $\theta_{n,p}=\bigO(\rho^n)$, we have
\[
\begin{split}
	\sup_{x\in \Reals} \EE|Z_k(x)|&\le C \Pr(|X_{k-1}-\EE(X_{k-1}\,|\, \xi_{k-1,k-m})|\ge \rho_1^m)\\
	&\quad +C\sup_{x\in \Reals} \Pr\left(\left\{\frac{X_{k-1}-x}{b}\in [-1,1]\right\}\right) \\
	&=\bigO(\rho/\rho_1)^m+ \bigO(\rho_1^m/b).
\end{split}
\]
Now, we can show the maximum of the skeleton process over $\{x_j\},j=1,\dots,q_n$ is small. Recall that $m$ is a polynomial of $n$, then we have

\[
\begin{split}
	&\Pr\left(\max_{1\le j\le q_n}\left| \sum_{k=1}^n Z_k(x_j)\right| \ge \sqrt{nb}(\log b^{-1})^{-2}\right)\\
	\le& q_n \frac{\max_{1\le j\le q_n}\EE\left|\sum_{k=1}^n Z_k(x_j)\right|}{\sqrt{nb}(\log b^{-1})^2}\\
	\le& \frac{nq_n}{\sqrt{nb}(\log b^{-1})^2}\sup_{x\in T} \EE|Z_k(x)|=o(1).
\end{split}
\]
Next, we show the third term of the decomposition of $K\left(\frac{X_{k-1}-x}{b}\right)$ in \cref{eq_decompose_K} can also be ignored. In order to show this, we define
\[
N_n(x)=\frac{1}{\sqrt{nb}}\sum_{k=1}^n \breve{\zeta}_k\EE\left[K\left(\frac{X_{k-1}-x}{b}\right)\,|\, \xi_{k-1,k-m}\right].
\]
Using the same argument as in \citep[Proof of Lemma 4.2]{Liu2010a}, we can approximate $N_n(x)$ by its skeleton process, since
$\sup_{x_{j-1}\le x\le x_j} \left| N_n(x)-N_n(x_j) \right| =o_{\Pr}(\log n)^{-2}$.
We first approximate $\sup_x |N_n(x)|$ by the maximum over the skeleton process. Then we have $\Pr\left(\max_{1\le j\le q_n} |N_n(x_j)|\ge (\log n)^{-2}\right)=o(1)$ using Freedman's inequality for martingale differences \cite{Freedman1975}.
Therefore, we can  approximate $M_n(x)$ by
\[
\frac{1}{nb}\sum_{k=1}^n \frac{\breve{\zeta}_k}{{\EE[\zeta_k^2]}} \left\{\EE\left[K\left(\frac{X_{k-1}-x}{b}\right)\,|\, \xi_{k-1,k-m}\right]- \EE\left[K\left(\frac{X_{k-1}-x}{b}\right)\,|\, \xi_{k-2,k-m}\right] \right\}.
\]
Furthermore, since {$|1-\EE[\breve{\zeta}_k^2]/\EE[\zeta_k^2]|=\bigO((\log n)^{-12/(p-2)})$}, we can replace {$\breve{\zeta}_k/\EE[\zeta_k^2]$} by $\breve{\zeta}_k/\breve{\sigma}^2$, which leads to the definition of $\tilde{M}_n(x)$. Therefore, we have proved
\[
\Pr\left(\sqrt{nb}\sup_{x\in T}\left| {M}_n(x) -\tilde{M}_n(x) \right|\ge (\log n)^{-2}\right)=o(1).
\]
{	Therefore, in order to finish the proof of \cref{lemma_approximate_Mn}, it suffices to show
	\[
	\Pr\left(\sqrt{nb} \sup_{x \in T} \left| R_n(x)\right|\ge (\log n)^{-2}\right)=o(1),
	\]
	where
	$R_n(x):=
	\frac{1}{nb}\sum_{j\in \cup_{i=1}^{\iota_n+1} I_i} u_j(x)$.
	Following the same argument as above using skeleton process, we only need to consider the grids $\{x_j, j=0,\dots,q_n\}$. Using the fact that $\tau<\tau_1$ and $n^{-\delta_1}=\bigO(b)$, again by Freedman's inequality for martingale differences, for some constant $C$ that
	\[
	&\Pr\left(\sqrt{nb} \sup_{0\le j\le q_n} \left| R_n(x_j)\right|\ge (\log n)^{-3}\right)\\
	&\le 4q_n\exp\left[\frac{(\log n)^{-6}nb}{-2C(\log n)^{-3}\sqrt{nb}-2C(n^{1-\tau_1+\tau}+n^{\tau_1})b}\right]=o(1),
	\]
	which finishes the proof of \cref{lemma_approximate_Mn}.
}

Observing that, since $K(\cdot)$ is supported on $[0,1]$,
one of the following two terms must be zero:
\[
\begin{split}
\EE\left[K\left(\frac{X_{k-1}-x}{b}\right)\,|\, \xi_{k-1,k-m}\right]- \EE\left[K\left(\frac{X_{k-1}-x}{b}\right)\,|\, \xi_{k-2,k-m}\right],\\
\EE\left[K^*\left(\frac{X_{k-1}-x}{b}\right)\,|\, \xi_{k-1,k-m}\right]- \EE\left[K^*\left(\frac{X_{k-1}-x}{b}\right)\,|\, \xi_{k-2,k-m}\right].
\end{split}
\]
Hence, defining $\widetilde{M}_n^*(x)$ similarly as $\widetilde{M}_n(x)$ using $K^*(\cdot)$ instead of $K(\cdot)$, by \cref{lemma_approximate_Mn}, we only need to focus on the following term

\[
\begin{split}
\hat{M}_n(x):&=\sqrt{\frac{nb}{2\lambda_K f(x)}}\left[\widetilde{M}_n(x)-\widetilde{M}_n^*(x) \right]\\
&=\frac{1}{\sqrt{nb\lambda_{\tilde{K}}f(x)}}\sum_{k\in \cup_{i=1}^{\iota_n} H_i} \frac{\breve{\zeta}_k}{\breve{\sigma}^2} \left\{\EE\left[\tilde{K}\left(\frac{X_{k-1}-x}{b}\right)\,|\, \xi_{k-1,k-m}\right]\right.\\
&\quad -\left.\EE\left[\tilde{K}\left(\frac{X_{k-1}-x}{b}\right)\,|\, \xi_{k-2,k-m}\right] \right\}.
\end{split}
\]
Clearly, in order to complete the proof of \cref{thm_main}, it suffices to show
\[\label{lemma_Gumbel}
\Pr\left(\sup_{x\in T}\left|\hat{M}_n(x)\right|-d_n\le \frac{z}{(2\log \bar{b}^{-1})^{1/2}}\right)\to e^{-2e^{-z}}.
\]

\subsubsection{Asymptomatic covariance structure}\label{proof_step4}

Next, we prove some results on the asymptomatic covariance structure of $\{\hat{M}_n(x)\}$ which will be needed later for Gaussian approximation using the results in \cite{Bickel1973}. Define the following quantities: $r(s):=\int K(x)K(x+s)\dee x/\lambda_K$,  $\hat{r}(s):=\EE\hat{M}_n(x)\hat{M}_n(x+s)$, $\tilde{r}(s):=\int \tilde{K}(x)\tilde{K}(x+s)\dee x/\lambda_{\tilde{K}}$, and $\tilde{K}_2:=\int_{-1}^{1}(\tilde{K}'(x))^2\dee x/ (2\lambda_{\tilde{K}})$.
Note that since $\tilde{K}'(0)>0$, we have $\int K(u)K^*(u\pm s)\dee u =\bigO(\int_{0}^{|s|} x(|s|-x)\dee x)=\bigO(|s|^3)=o(|s|^2)$. Then by the definition of $\tilde{r}(s)$, using $\lambda_{\tilde{K}}=2\lambda_K$, we have
\[
\tilde{r}(s)&=\int \tilde{K}(v)\tilde{K}(v+s)\dee v/\lambda_{\tilde{K}}\\
&=\frac{1}{\lambda_{\tilde{K}}}\int \left[K(v)-K^*(v)\right]\left[K(v+s)-K^*(v+s)\right]\dee x\\
&=\frac{1}{2\lambda_K}\left[
\int K(v+s)K(v)\dee v
+\int K^*(v+s)K^*(v)\dee v\right.\\
&\quad\left.
-\int K^*(v+s)K(v)\dee v
-\int K(v+s)K^*(v)\dee v
\right]\\
&= r(s)+o(|s|^2).
\]

Next, according to \citep[Theorems B1 and B2]{Bickel1973}, we have
$r(s)=1-K_2|s|^2+o(|s|^2)$.
Note that
\[
\tilde{K}_2=\int_{-1}^{1}(\tilde{K}'(x))^2\dee x/ (2\lambda_{\tilde{K}})=\frac{1}{2}\int_{-1}^{1}(\tilde{K}'(x))^2\dee x/ (2\lambda_K)
=\frac{1}{2} (2K_2)=K_2.
\]
This implies 
$\tilde{r}(s)=1-\tilde{K}_2|s|^2+o(|s|^2)$, which can also be obtained directly from \citep[Theorems B1 and B2]{Bickel1973}.

Next, we show
$\hat{r}(s)=\tilde{r}(s)+\bigO(b)$.
Note that $\{\breve{\zeta}_k\}$ are {uncorrelated} and $\EE{\breve{\zeta}_k}=0$. Then, using 
$|f(v+s)-\sqrt{f(t)f(s)}|=\bigO(b)$
uniformly over $|s-t|\le 2b$ and $|v|\le 2b$, we have
\[
\begin{split}
	&\EE\hat{M}_n(t)\hat{M}_n{(s)}\\
	=&\frac{1}{nb\lambda_{\tilde{K}}}\int \frac{1}{\sqrt{f(t)f(s)}}
	\sum_{k\in \cup_{i=1}^{\iota_n}H_i}\left\{\EE\left[\tilde{K}\left(\frac{X_{k-1}-t}{b}\right)\tilde{K}\left(\frac{X_{k-1}-s}{b}\right) \right]+\bigO(b^2)\right\}\\
	=&\frac{1}{b\lambda_{\tilde{K}}}\int \frac{1}{f(v+s)+\bigO(b)}\tilde{K}\left(\frac{v-t+s}{b}\right)\tilde{K}\left(\frac{v}{b}\right)f(v+s)\dee v + \bigO(b)\\
	=&\frac{1}{\lambda_{\tilde{K}}}\int \tilde{K}\left(v-t+s\right)\tilde{K}\left(v\right)\dee v + \bigO(b)=\tilde{r}(t-s)+\bigO(b).
\end{split}
\]
Therefore, we have proved that, as $s\to 0$, 
	\[\label{eq_M_n_covariance}
	\tilde{r}(s)=1-\tilde{K}_2|s|^2+o(|s|^2),\quad \tilde{r}(s)=r(s)+o(|s|^2),\quad \hat{r}(s)=\tilde{r}(s)+\bigO(b).
	\]

\subsubsection{Gaussian approximation}\label{proof_step5}

Now, we go back to prove \cref{lemma_Gumbel}.
We use similar techniques as in \citep[Proof of Lemma 4.5]{Liu2010a}. First, as in \cite{Bickel1973}, we split the interval $T$ into alternating big and small intervals $W_1,V_1,\dots,W_N,V_N$, where $W_i=[a_i,a_i+w]$, $V_i=[a_i+w,a_{i+1}]$, $a_i=(i-1)(w+v)$, and $N=\lfloor (u-l)/(w+v)\rfloor$. We let $w$ be fixed, and $v$ be small which goes to $0$. {Since $u$ and $l$ are fixed numbers, without loss of generality, we assume $l=0$ and $u=1$ in this proof.}

Next, we approximate
$\Omega^{+}:=\sup_{0\le t\le 1} \hat{M}_n(t)$
by big blocks $\{W_k\}$. That is, by
$\Psi^{+}:=\max_{1\le k\le N} \Upsilon^{+}_k$, where $\Upsilon^{+}:=\sup_{t\in W_k} \hat{M}_n(t)$.
Then we further approximate $\Upsilon^{+}_k$ via discretization by
$\Xi_k^+:=\max_{1\le j\le \chi} \hat{M}_n(a_k+jax^{-1})$, where $\chi=\lfloor w x/a\rfloor$ with $a>0$. We define $\Omega^-$, $\Psi^-$, $\Upsilon^-_k$, and $\Xi_k^-$ similarly by replacing $\sup$ or $\max$ by $\inf$ or $\min$, respectively. Letting $\Omega=\max(\Omega^+, -\Omega^-)=\sup_{0\le t\le 1}|\hat{M}_n(t)|$ and $x_z=d_n+z/(2\log b^{-1})^{1/2}$, we have

\[
\begin{split}
&\left|\Pr(\Omega\ge x_z)-\Pr(\{\Psi^+\ge x_z\}\cup\{\Psi^-\le -x_z\})\right|\le R_1+R_2,\\
&\left|\Pr(\{\Psi^+\ge x_z\}\cup\{\Psi^-\le -x_z\})-\Pr\left(\bigcup_{k=1}^N\left\{\Xi_k^+\ge x_z\right\} \cup \bigcup_{k=1}^N\left\{\Xi_k^-\le -x_z\right\} \right)\right|\\
& \le R_3+R_4.
\end{split}
\]
where
\[
R_1:&=\Pr\left(\max_{1\le k\le N} \sup_{t\in V_k} \hat{M}_n(t)\ge x_z\right),\\
R_2:&=\Pr\left(\min_{1\le k\le N} \inf_{t\in V_k} \hat{M}_n(t)\le -x_z\right),\\
R_3:&=\sum_{k=1}^N \left|\Pr(\Upsilon_k^+\ge x_z)-\Pr(\Xi_k^+\ge x_z)\right|,\\
R_4:&=\sum_{k=1}^N \left|\Pr(\Upsilon_k^-\le -x_z)-\Pr(\Xi_k^-\le -x_z)\right|.
\]

Next, we are ready to apply Gaussian approximation. 
	We first use discretization for approximating $\hat{M}_n(x)$.
	Let $s_j=j/(\log n)^6, 1\le j<t_n$, where $t_n=1+\lfloor (\log n)^6t\rfloor$, $s_{t_n}=t$. Write $[s_{j-1},s_j]=\bigcup_{k=1}^{q_n}[s_{j,k-1},s_{j,k}]$, where $q_n=\lfloor (s_j-s_{j-1})n^2\rfloor =\lfloor n^2/(\log n)^6 \rfloor$ and $s_{j,k}-s_{j,k-1}=(s_j-s_{j-1})/q_n$. Following the same arguments as in \citep[Proof of Lemma 4.6]{Liu2010a}, we have the following discretization approximation holds for all large enough $Q$,
	\[
	&\Pr\left(\sup_{0\le s\le t} \hat{M}_n(v+s)\ge x \right)\\
	\le& \Pr\left(\max_{1\le j\le t_n}\hat{M}_n(v+s_j)\ge x-(\log n)^{-2}\right)+Cn^{-Q}.
	\]
	Next, we apply the multivariate Gaussian approximation from \cite{Zaitsev1987}. To this end, similar to the definition of $u_j(t)$, we first define
	\[
	\tilde{u}_j(t):&= \sum_{k\in H_j}\frac{\breve{\zeta}_k}{\breve{\sigma}^2} \left\{\EE\left[\tilde{K}\left(\frac{X_{k-1}-t}{b}\right)\,|\, \xi_{k-1,k-m}\right]\right.\\
	&\quad -\left.\EE\left[\tilde{K}\left(\frac{X_{k-1}-t}{b}\right)\,|\, \xi_{k-2,k-m}\right] \right\}, \quad j=1,\dots, \iota_n
	\]
	{Note that the sequence of random variables $\{\tilde{u}_j(t), j=1,\cdots,\iota_n\}$ are independent.}
	Then we define
	\[
	\begin{split}
		\widehat{u}_j(t)&:=\tilde{u}_j(t)\mathbbm{1}\{|\tilde{u}_j(t)|\le \sqrt{nb}(\log n)^{-20p/(p-2)}\}\\
		&\quad -\EE\left[\tilde{u}_j(t)\mathbbm{1}\{|\tilde{u}_j(t)|\le \sqrt{nb}(\log n)^{-20p/(p-2)}\}\right].
	\end{split}
	\]
	Now we introduce $\widehat{M}_n(t):=\frac{1}{\sqrt{nb\lambda_{\tilde{K}}f(t)}}\sum_{j=1}^{\iota_n} \widehat{u}_j(t)$.
	Then using \citep[Theorem 1.1]{Zaitsev1987} as well as $\sup_t \max_{1\le j\le \iota_n}\|\widehat{u}_j(t)-\tilde{u}_j(t)\|\le Cn^{-Q}$ for large enough $Q$, we have 
	\[
	\begin{split}
		&\Pr\left(\max_{1\le j\le t_n} \hat{M}_n(v+s_j)\ge x-(\log n)^{-2} \right)\\
		\le& \Pr\left(\max_{1\le j\le t_n} \widehat{M}_n(v+s_j)\ge x-(\log n)^{-2}\right)+Cn^{-Q}\\
		\le& {\Pr\left(\max_{1\le j\le t_n} Y_n(j)\ge x-2(\log n)^{-2}\right)+Ct_n^{5/2}\exp\left(-\frac{C(\log n)^{18p/(p-1)}}{t_n^{5/2}}\right)}\\
		&\quad +Cn^{-Q},
	\end{split}
	\]
	where $(Y_n(1),\dots,Y_n(t_n))$ is a centered Gaussian random vector with covariance matrix
	$\widehat{\Sigma}_n=\Cov(\widehat{M}_n(v+s_1),\dots,\widehat{M}_n(v+s_{t_n}))$.
	
	Let $\psi$ be the density function of standard Gaussian, and $H_2(a)$ be the Pickands constants \citep[Theorem A1, Lemma A1, and Lemma A3]{Bickel1973}. Using \cref{eq_M_n_covariance}, let $t>0$ be such that $\inf \{s^{-2}(1-\tilde{r}(s)): 0\le s\le t\}>0$. 
	Following exactly the arguments in \citep[Proof of Lemma 4.6]{Liu2010a} to apply \citep[Lemma A3 and Lemma A4]{Bickel1973}, we can obtain that for $a>0$,
	\[
	\Pr\left(\bigcup_{j=1}^{\lfloor tx/a\rfloor}\left\{\hat{M}_n(v+jax^{-1})\ge x\right\}\right)=x\psi(x)\frac{H_2(a)}{a}\tilde{K}_2^{1/2}t+o(x\psi(x)),
	\]
	uniformly over $0\le v\le 1$. The limit when $a\to 0$ also holds, that is
	\[\label{lemma_key1}
	\Pr\left(\bigcup_{0\le s\le t}\left\{\hat{M}_n(v+s)\ge  x\right\}\right)=x\psi(x)\tilde{K}_2^{1/2}t/\sqrt{\pi} +o(x\psi(x)),
	\]
	where we have used the Pickands constants $H_2=\lim_{a\to 0} H_2(a)/a=1/\sqrt{\pi}$.
	The left tail version of the tail bounds also holds with $\ge x$ replaced by $\le x$.
	Furthermore, we can show through elementary calculations that 
\[
\lim_{a\to 0}\limsup_{v\to 0}\limsup_{n\to \infty} R_j=0, \quad j=1,\dots,4.
\]
Therefore, it suffices to show the following convergence to Gumbel law	
\[\label{lemma_key2}
\lim_{a\to 0}\limsup_{v\to 0}\limsup_{n\to\infty}\left|\Pr\left(\bigcup_{k=1}^N\left\{\Xi_k^+\ge x_z\right\} \cup \bigcup_{k=1}^N\left\{\Xi_k^-\le -x_z\right\} \right)-(1-e^{-2e^{-z}})\right|=0.
\]

\subsubsection{Convergence to Gumbel distribution}\label{proof_step6}

The main steps of the proof of \cref{lemma_key2} are as follows. First, we approximate $\hat{M}_n(t)$ by $Y_n(t)$. Then, we approximate $Y_n(t)$ by another quantity $\hat{M}_n'(t)$ which is defined similarly to $\hat{M}_n(x)$ but using a sequence of i.i.d.\ random variables instead of the dependent time series $\{X_k\}$. Finally, we apply \citep[Theorem]{Rosenblatt1976} to show convergence to Gumbel distribution. 

We define
\[
\begin{split}
	\mathbf{B}_{k,j}:&=\{\hat{M}_n(a_k+jax^{-1})\ge x\}\cup \{\hat{M}_n(a_k+jax^{-1})\le -x\},\\
	\mathbf{D}_{k,j}:&=\{Y_n(a_k+jax^{-1})\ge x\}\cup \{Y_n(a_k+jax^{-1})\le -x\},
\end{split}
\]
where $Y_n(\cdot)$ is a centered Gaussian process with covariance function
\[
\Cov(Y_n(s_1),Y_n(s_2))=\Cov(\hat{M}_n(s_1),\hat{M}_n(s_2)).
\]

First we approximate $\hat{M}_n(t)$ using $Y_n(t)$. Recall that $w$ and $v$ are the lengths of big and small blocks $W_i$ and $V_i$. Let $N=\lfloor 1/(w+v)\rfloor$. 	Define a different truncation order for $M_n(t)$ by $\widehat{M}'_n(t):=\frac{1}{\sqrt{nb\lambda_{\tilde{K}}f(t)}}\sum_{j=1}^{\iota_n} \widehat{u}'_j(t)$ for given $d$, where
\[
\begin{split}
	\widehat{u}'_j(t)&:=\tilde{u}_j(t)\mathbbm{1}\{|\tilde{u}_j(t)|\le \sqrt{nb}(\log n)^{-20dp/(p-2)}\}\\
	&\quad-\EE\left[ \tilde{u}_j(t)\mathbbm{1}\{|\tilde{u}_j(t)|\le \sqrt{nb}(\log n)^{-20dp/(p-2)}\}\right].
\end{split}
\] 
Then using $\widehat{M}'_n(t)$ and following exactly the same proof from \citep[Proof of Lemma 4.10]{Liu2010a} to get that, for any fixed integer $l$ that $1\le l\le N/2$,
	\[\label{temp111}
	\begin{split}
	&\left|\Pr\left(\bigcup_{k=1}^N \mathbf{A}_k\right) -\sum_{d=1}^{2l-1} (-1)^{d-1}\left(\sum_{1\le i_1<\cdots< i_d\le N}-\sum_{\mathcal{I}}\right) \Pr\left(\bigcap_{j=1}^d \mathbf{C}_{i_j}\right)  \right|\\
	\le& \frac{C^{2l}}{(2l)!}+\bigO\left(\frac{1}{\log n}\right),
	\end{split}
	\]
	where 
	$\mathbf{A}_k:=\bigcup_{j=1}^{\lfloor wx/a \rfloor} \mathbf{B}_{k,j}$, $\mathbf{C}_k:=\bigcup_{j=1}^{\lfloor wx/a \rfloor} \mathbf{D}_{k,j}$,
	$C$ does not depend on $l$, and
	\[
	\mathcal{I}:=\left\{
	1\le i_1<\cdots< i_d\le N: \min_{1\le j\le d-1} q_j\le \lfloor 2w^{-1}+2\rfloor
	\right\}.
	\] 

Next, we construct $\hat{M}_n'(t)$ in the following way. Let $\{\eta_i^{(k)}\}, i\le k\le n,$ be i.i.d.\ copies of $\{\eta_i\}$, and $\xi_j^{(k)}=(\dots,\eta_{j-1}^{(k)},\eta_j^{(k)})$. Let $X_i^{(k)}=G(\xi_j^{(k)})$. Note that $X_k^{(k)}, 1\le k\le n,$ are i.i.d. Now define $\mathbf{A}_k'$ the same as $\mathbf{A}_k$ except by replacing $Y_j$ and $\{\eta_i\}$ with $X_k^{(k)}$ and $\{\eta_i^{(k)}\}$, respectively. Repeat the previous arguments for getting \cref{temp111}, we have
\[
\begin{split}
&\left|\Pr\left(\bigcup_{k=1}^N \mathbf{A}_k'\right) -\sum_{d=1}^{2l-1} (-1)^{d-1}\left(\sum_{1\le i_1<\cdots< i_d\le N}-\sum_{\mathcal{I}}\right) \Pr\left(\bigcap_{j=1}^d \mathbf{C}_{i_j}\right)  \right|\\
\le& \frac{C^{2l}}{(2l)!}+\bigO\left(\frac{1}{\log n}\right).
\end{split}
\]
Letting $n\to \infty$ then $l\to \infty$, by triangle inequality, we have
\[\limsup_{n\to \infty}\left| \Pr\left(\bigcup_{k=1}^N \mathbf{A}_k\right)-\Pr\left(\bigcup_{k=1}^N \mathbf{A}_k'\right)    \right|=0.
\]
Now the key observation is that we can deal with $\{\mathbf{A}_k'\}$ now and $\mathbf{A}_k'$ are defined using $\{X_k^{(k)}\}$ which are i.i.d. Next, we define $R_1'$ to $R_4'$ the same as $R_1$ to $R_4$ except using $\{X_k^{(k)}\}$ and $\{\eta_i^{k}\}$ instead of $\{X_k\}$ and $\{\eta_i\}$, then by \cref{lemma_key1} and elementary calculations again we have $\lim_{a\to 0}\limsup_{v\to 0}\limsup_{n\to \infty} R_j'=0$ for $j=1,\dots,4$. This implies
\[
\lim_{a\to 0}\limsup_{v\to 0}\limsup_{n\to \infty}\left|\Pr\left(\bigcup_{k=1}^N \mathbf{A}_k'\right)- \Pr\left(\sup_{0\le t \le 1} \left| \hat{M}_n'(t)\right|<x \right) \right| =0,
\]
where $\hat{M}_n'(t)$ is defined in the same way as $\hat{M}_n(t)$ by replacing $\{X_k\}$ with $\{X_k^{(k)}\}$, and $\{\eta_i\}$ with $\{\eta_i^{(k)}\}$. Finally, since $\{X_k^{(k)}\}$ are i.i.d., we can apply \citep[Theorem]{Rosenblatt1976}, which leads to the convergence of $\Pr\left(\sup_{0\le t \le 1} \left| \hat{M}_n'(t)\right|<x_z \right)$ to $e^{-2e^{-z}}$. This completes the proof of \cref{thm_main}.

\subsection{Proof of \cref{thm_estimation}}\label{proof_thm_estimation}
First, let $r_n$ and $s_n$ be positive sequences, then $r_n=\Omega(s_n)$ if $s_n=o(r_n)$. On the other hand, $r_n=\Theta(s_n)$ if both $s_n=\bigO(r_n)$ and $r_n=\bigO(s_n)$ hold. 
Note that
\[
\begin{split}
&\Pr\left(\left\{\hat{M}=M\right\}\cap \left\{\max_{1\le i\le M} |\hat{x}_i-x_i|< c_n \right\}\right)\\
&=\Pr\left(\max_{1\le i\le M} |\hat{x}_i-x_i|< c_n \mid \hat{M}=M \right)\Pr\left(\hat{M}=M\right).
\end{split}
\]
{
	We first argue that $\Pr\left(\hat{M}<M\right)\to 0$, which implies at least one change point hasn't been detected, then we can write
	$$\Pr\left(\hat{M}<M\right)\le \sum_{i=1}^M \Pr\left(\textrm{the change point $a_i$ is not detected}\right).$$ Then, {by the validity of the bootstrap procedure, when $\sqrt{\frac{b\log n}{n}}=o(\tilde{\Delta}_n)$, the power of the test goes to $1$ as $n\to \infty$} which implies that for any $i$, 
	$$\Pr\left(\textrm{the change point $a_i$ is not detected}\right)\to 0.$$
	This conclude that $\Pr\left(\hat{M}<M\right)\to 0$.
}

{
	Next we argue that $\Pr\left(\hat{M}>M\right)\to \alpha$.
	Note that $\hat{M}>M$ implies there is a set $\tilde{T}$ without any change point in it, however,
	$ \sup_{x\in \tilde{T}}|t_n(x)|\ge C_{n,\alpha}$. Note that by our algorithm, we can consider $\tilde{T}$ to be the largest set constructed by ruling out $M$ intervals from $[l,u]$ such that each interval has length $2b$ and contains one change point. Then since $M$ is a fixed constant and $b\to 0$, we have $|\tilde{T}|= (|u-l|-2Mb)^+ \to |u-l|$. Then we can apply our main result \cref{thm_main} again on $\tilde{T}$ to get that 	$\Pr\left( \sup_{x\in \tilde{T}}|t_n(x)|\ge C_{n,\alpha}\right)\to \alpha$, which implies $\Pr\left(\hat{M}>M\right)\to \alpha$.
}

 Therefore, we have $\Pr\left(\hat{M}=M\right)\to 1-\alpha$. Then it suffices to show
\[
\Pr\left(\max_{1\le i\le M} |\hat{x}_i-x_i|< c_n \mid \hat{M}=M \right)\to 1.
\]
Since $M$ is finite, we only need to focus on one change point. Let $x_0$ be any of the true change point and $\hat{x}$ be its estimate, it suffices to show $\Pr\left(|\hat{x}-x_0|\ge c_n \mid \hat{M}=M\right)\to 0$.
Without loss of generality, we assume $\hat{x}-x_0=\hat{c}_n=o_{\Pr}(b)$ and $t_n(x_0)>0$. The case $t_n(x_0)<0$ can be shown using similar arguments. Now we follow similar arguments as in \cite{Muller1992}. Define
$\zeta(c):=t_n(x_0+c)
-t_n(x_0)$,
for $c=o(b)$.
Then we can write
$\hat{c}_n=\arg\max \zeta(c)$.
Therefore, it suffices to {show 
$\hat{c}_n=\bigO_{\Pr}\left(\frac{1}{\tilde{\Delta}_n}\sqrt{\frac{b\log n}{n}}\right)$}.
Suppose $b$ is small enough such that the $b$-neighborhood of $x_0$ does not include any other change points, then we apply the previous decomposition in \cref{eq_decomposition}. Note that since $x_0$ is a change point, without loss of generality, we assume $\mu(x)$ is left continuous at $x=x_0$, then the following term has the order of  {$\Theta(\tilde{\Delta}_n)$}:
\[
\frac{1}{nb}\left|\sum_{k=1}^n \tK\left(\frac{X_{k-1}-x_0}{b}\right)\frac{\mu(x_0)}{f(x_0)}-\sum_{k=1}^n \tK\left(\frac{X_{k-1}-(x_0+c)}{b}\right)\frac{\mu(x_0+c)}{f(x_0+c)}\right|.
\]
Furthermore, using $\int_{0}^s K(x)\dee x = \Theta(s^2)$ because of $\tilde{K}'(0)>0$, considering cases $|X_{k-1}-x_0| \in [0,c]$ and $|X_{k-1}-x_0|\in (c,b]$ separately, we have
\[
\begin{split}
&\left|\frac{1}{nb f(x_0+c)}\sum_{k=1}^n \tK\left(\frac{X_{k-1}-(x_0+c)}{b}\right)
\left[ \mu(X_{k-1})-\mu(x_0+c)\right]\right.\\
&\quad -\left.\frac{1}{nb f(x_0)}\sum_{k=1}^n \tK\left(\frac{X_{k-1}-x_0}{b}\right)
\left[ \mu(X_{k-1})-\mu(x_0)\right]\right|\\
=&\left[\frac{1}{b }\int_{0}^{c}K\left(\frac{x}{b}\right)\dee x\right] {\Theta_{\Pr}(\tilde{\Delta}_n)}\\
&\quad + \left[\frac{1}{b }\int_{c}^{b}K\left(\frac{x}{b}\right)\dee x\right] \bigO_{\Pr}(b^3+\sqrt{b\log n/n})\\
=&\left[\int_{0}^{c/b}K(x)\dee x\right]{\Theta_{\Pr}(\tilde{\Delta}_n)} +\bigO_{\Pr}(b^3+\sqrt{b\log n/n})\\
=&{(c/b)^2\Theta_{\Pr}\left(\tilde{\Delta}_n\right)}+\bigO_{\Pr}(b^3+\sqrt{b\log n/n}).
\end{split}
\]
Finally, {by the assumptions on $K'$ in \cref{thm_estimation}, we can follow the same arguments in the proof of \cref{thm_main} as the $m$-dependent approximation \cref{proof_step2} and alternating big/small blocks \cref{proof_step3} applying to $\tilde{K}'$ instead of $\tilde{K}$ to get
\[
\frac{1}{nb f(x)}\sum_{k=1}^n\tK'\left(\frac{X_{k-1}-x}{b}\right)\epsilon_k=\bigO_{\Pr}\left(\sqrt{\frac{\log n}{nb}}\right)
\]
Furthermore, using the fact that $|\tilde{K}''(u)|$ is uniformly bounded and mean value theorem, we have
\[
&\EE\left[\frac{1}{f(x)}\left(\tK\left(\frac{X_{k-1}-x}{b}\right)-\tK\left(\frac{X_{k-1}-x+c}{b}\right)+\left(\frac{c}{b}\right)\tK'\left(\frac{X_{k-1}-x}{b}\right)\right)^2\right] \\
&=\int \frac{1}{f(x)}\left[\tK\left(\frac{y-x}{b}\right)-\tK\left(\frac{y-x+c}{b}\right)+\left(\frac{c}{b}\right)\tK'\left(\frac{y-x}{b}\right)\right]^2 f(y)\dee y\\
&=\int b\left[\tK\left(t\right)-\tK\left(t+\frac{c}{b}\right)+\left(\frac{c}{b}\right)\tK'\left(t\right)\right]^2 \frac{f(tb+x)}{f(x)}\dee t\\
&=b\left(\frac{c}{b}\right)^2\int \left[\frac{\tK(t)-\tK(t+c/b)}{c/b}+\tK'(t)\right]^2(1+\bigO(b))\dee t=\bigO\left(b\left(\frac{c}{b}\right)^4\right).
\]
Next, we define a new kernel $\check{K}$ such that 
\begin{align*}
&\check{K}\left(\frac{X_{k-1}-x}{b}\right)\\
:=&\left(\frac{b}{c}\right)^2\left[\tK\left(\frac{X_{k-1}-x}{b}\right)-\tK\left(\frac{X_{k-1}-x+c}{b}\right)+\left(\frac{c}{b}\right)\tK'\left(\frac{X_{k-1}-x}{b}\right)\right]
\end{align*}
so we have
$\EE\left[\frac{1}{f(x)}\check{K}\left(\frac{X_{k-1}-x}{b}\right)^2\right]=\bigO(b)$.
Then we can approximate the following term using the same arguments of $m$-dependent approximation and alternating big/small blocks as in \cref{proof_step2} and \cref{proof_step3} in the proof of \cref{thm_main} applying to this new kernel $\check{K}$ to get
\[
&\frac{1}{nb f(x)}\sum_{k=1}^n \left[\tK\left(\frac{X_{k-1}-x}{b}\right)-\tK\left(\frac{X_{k-1}-x+c}{b}\right)+\left(\frac{c}{b}\right)\tK'\left(\frac{X_{k-1}-x}{b}\right)\right]{\epsilon_k}\\
&=\left(\frac{c}{b}\right)^2\left[\frac{1}{nb f(x)}\sum_{k=1}^n\check{K}\left(\frac{X_{k-1}-x}{b}\right)\epsilon_k\right]=\bigO_{\Pr}\left(\left(\frac{c}{b}\right)^2\sqrt{\frac{\log n}{nb}}\right).
\]
} 
Therefore, we have
\[
&\frac{1}{nb f(x)}\sum_{k=1}^n \left[\tK\left(\frac{X_{k-1}-x}{b}\right)-\tK\left(\frac{X_{k-1}-x+c}{b}\right)\right]{\epsilon_k}\\
&=\left(\frac{c}{b}\right)\left[\frac{-1}{nb f(x)}\sum_{k=1}^n\tK'\left(\frac{X_{k-1}-x}{b}\right)\epsilon_k\right]+{\bigO_{\Pr}\left(\left(\frac{c}{b}\right)^2\sqrt{\frac{\log n}{nb}}\right)}\\
&=\left(\frac{c}{b}\right){\bigO_{\Pr}\left(\sqrt{\frac{\log n}{nb}}\right)}+\left(\frac{c}{b}\right)^2 {\bigO_{\Pr}\left(\sqrt{\frac{\log n}{nb}}\right)}.
\]
Then using {$\sqrt{\frac{\log n}{nb}}=o(\tilde{\Delta}_n)$} we can conclude that
\[
\zeta(c)=-\left(\frac{c}{b}\right)^2{\Theta_{\Pr}\left(\tilde{\Delta}_n\right)}+\left(\frac{c}{b}\right){\bigO_{\Pr}\left(\sqrt{\frac{\log n}{nb}}\right)}-\bigO_{\Pr}(b^3+\sqrt{b \log n/n}).
\]
Recall that the estimated change point $\hat{x}=x_0+\hat{c}_n$, where 
$\hat{c}_n=\arg\max \zeta(c)$,
then we have
\[
\hat{c}_n={\bigO_{\Pr}\left(\frac{b}{\tilde{\Delta}_n}\sqrt{\frac{\log n}{nb}}\right)=\bigO_{\Pr}\left(\frac{1}{\tilde{\Delta}_n}\sqrt{\frac{b\log n}{n}}\right)},
\]
whenever {$b^4=o((\log n)/(n\tilde{\Delta}_n))$ and $b^3=o((\log n)/(n\tilde{\Delta}_n^2))$}. This is always true since we have assumed {$\delta_2\le 1/4$} which implies {$b=\bigO(n^{-1/4})$} so {$b^4=\bigO(1/n)=o((\log n)/n)$}.
Therefore, if we choose $c_n>0$ such that $\hat{c}_n=o(c_n)$,
then we have $\Pr(|\hat{c}_n|< c_n)\to 0$, which implies
$\Pr\left(|\hat{x}-x_0|\ge c_n\mid \hat{M}=M\right)\to 0$.

\section{Additional proofs}
\subsection{Proof of 
	\cref{remark1}}\label{proof_remark1}
	For $\sigma_n^2(x)$, we first write it as the sum of three terms:
	\[
	\sigma_n^2(x)&=\frac{1}{nhf_n(x)}\sum_{k=1}^n W\left(\frac{X_k-x}{h}\right)\epsilon_k^2\\
	&\quad+\frac{2}{nhf_n(x)}\sum_{k=1}^n W\left(\frac{X_k-x}{h}\right)[\mu(X_k)-\mu_n(X_k)]\epsilon_k\\
	&\quad+\frac{1}{nhf_n(x)}\sum_{k=1}^n W\left(\frac{X_k-x}{h}\right)[\mu(X_k)-\mu_n(X_k)]^2.
	\]
	For the first term, we first approximate $\epsilon_k^2$ by $\{\EE[\epsilon_k^2\mid \xi_{k,k-m}]\}$ where $m=\lfloor n^{\tau}\rfloor$ with $\tau>0$ small enough. Using the same argument as in \cref{proof_thm_main}, we have
	\[
	\sup_x\left|\frac{1}{nh}\sum_{k=1}^n W\left(\frac{X_k-x}{h}\right)\left\{\epsilon_k^2-\EE[\epsilon_k^2\mid \xi_{k,k-m}]\right\}\right|=\bigO_{\Pr}\left(\rho^m\right)={o_{\Pr}\left(n^{-1/2}\right)},
	\]
	{where we choose $m=c\log n$ with $c>-\frac{1}{2\log(\rho)}$. We then divide $1,\dots,m$ into $\lfloor n/m\rfloor+1$ blocks indexed by $1,\cdots, \lfloor n/m\rfloor+1$. Then it's clear that the sum of blocks with odd indices is independent with the sum of blocks with even indices.}
	Following the same argument as the proof of \citep[Theorem 2.5]{Liu2010a} {for each subsequence of the blocks}, {and use a union bound,} we can get
	\[
	&\sup_x\left|\frac{1}{nhf_n(x)}\sum_{k=1}^n W\left(\frac{X_k-x}{h}\right)\EE[\epsilon_k^2\mid \xi_{k,k-m}]\right|\\
	&\quad={\bigO_{\Pr}\left(2m\left(h^2+\sqrt{\frac{\log (n/m)}{(n/m)h}}\right)\right)}\\
	&\quad={\bigO_{\Pr}\left(\log n\left(h^2+\frac{\log n}{\sqrt{nh}}\right)\right)=\bigO_{\Pr}\left(h^2\log n + \frac{(\log n)^2}{\sqrt{nh}}\right)}
	\]
	
	For the second term, we first approximate $\{\epsilon_k\}$ using $\{\epsilon_k'\}$, where $\epsilon_k':=\EE[\epsilon_k\mid \xi_{k,k-m}]-\EE[\epsilon_k\mid \xi_{k-1,k-m}]$. Then following the same argument as in \cref{proof_thm_main} we have
	
		\[
	\sup_x\left|\frac{1}{nh}\sum_{k=1}^n W\left(\frac{X_k-x}{h}\right)\left(\epsilon_k-\epsilon_k'\right)\right|=\bigO\left(\rho^m\right).
	\]
	Then, {again choosing $m=c\log n$ and divide $1,\cdots,n$ into $\lfloor n/m\rfloor+1$ blocks,} by the same argument as in \citep[pp.~1875]{Zhao2008}, we can get
	\[
	&\sup_x\left|\frac{2}{nhf_n(x)}\sum_{k=1}^n W\left(\frac{X_k-x}{h}\right)[\mu(X_k)-\mu_n(X_k)]\epsilon_k\right|\\
	&\quad =\bigO_{\Pr}\left({(\log n)^2}\left(\frac{\log n}{nh^{5/2}}+\rho^m\right)\right)={\bigO_{\Pr}\left(\frac{(\log n)^3}{nh^{5/2}}+\frac{(\log n)^2}{\sqrt{n}}\right)}
	\]
	Finally, for the last term, we have
	\[
	&\sup_x \left|\frac{1}{nhf_n(x)}\sum_{k=1}^n W\left(\frac{X_k-x}{h}\right)[\mu(X_k)-\mu_n(X_k)]^2\right|\\
	&=\bigO_{\Pr}\left(h^4+\frac{\log n}{nh}\right)\cdot \sup_x \frac{1}{nh}\sum_{k=1}^n\left|W\left(\frac{X_k-x}{h}\right)\right|=\bigO_{\Pr}\left(h^4+\frac{\log n}{nh}\right).
	\]	
	{Then, using $0<\delta_1<1/4$ we have that
	\[
	\sup_x \left| \sigma_n^2(x)-\sigma^2(x)\right|&=\bigO_{\Pr}\left(h^2\log n+\frac{(\log n)^2}{\sqrt{nh}}+\frac{(\log n)^3}{nh^{5/2}}\right)
	\\
	&=\bigO_{\Pr}\left(h^2\log n+\frac{(\log n)^3}{\sqrt{nh}}\right)
	\]
	}
	For $f_n(x)$, similarly, by the same arguments as the proof for $\sigma_n^2(x)$, following the proof of \citep[Lemma 4.4]{Liu2010a}, we can obtain
{$\sup_x \left|f_n(x)-f(x)\right|=\bigO_{\Pr}\left(\frac{(\log n)^3}{\sqrt{nh}}+h^2\log n\right)$.}

\subsection{Proof of \cref{prop}}
Since $\{U_k\}_{k=0}^n$ are i.i.d. standard Gaussian distributed random variables, the proof for this proposition is simpler than \cref{thm_main}. We can immediately prove the convergence to Gumbel distribution by using \citep[Theorem 1]{Rosenblatt1976}.